\documentclass[11pt]{article}

\usepackage{amsmath,amsthm} 
\usepackage{amssymb,mathrsfs} 
\usepackage{a4wide} 
\usepackage{graphicx} 
\usepackage{color}
\usepackage{enumerate}  
\usepackage{hyperref}
\usepackage{fullpage}
\usepackage{epsfig}
\usepackage{caption}
\usepackage{subcaption}
\usepackage{dsfont}

\newtheorem{theorem}{Theorem}[section]

\newtheorem{lemma}[theorem]{Lemma}
\newtheorem{proposition}[theorem]{Proposition}
\newtheorem{corollary}[theorem]{Corollary}
\newtheorem{remark}[theorem]{Remark}

\newtheorem{assumption}[theorem]{Assumption}

\newenvironment{example small}[1][Example]{\begin{trivlist}
\item[\hskip \labelsep {\bfseries #1}]}{\end{trivlist}}
\newenvironment{remark small}[1][Remark]{\begin{trivlist}
\item[\hskip \labelsep {\bfseries #1}]}{\end{trivlist}}

\newcommand{\R}{\mathbb{R}}

\newcommand{\E}{\mathbb{E}}
\newcommand{\El}{\mathcal{E}}
\newcommand{\de}{\Delta}

\newcommand{\e}{\varepsilon}

\newcommand{\n}{\nabla}

\newcommand{\Lop}{\mathcal{L}}

\newcommand{\1}{\mathds{1}}
\newcommand{\K}{\mathcal{K}}

\newcommand{\Lopgen}{\mathcal{L}}

\newcommand{\U}{U}

\newcommand{\er}{v_{\mathrm{min}}}
\newcommand{\ef}{v_{\mathrm{max}}}

\newcommand{\ds}{\displaystyle}

\newcommand{\vp}{\varphi}
\newcommand{\UO}{U_{\rm org}}
\newcommand{\UN}{U_{\rm new}}

\renewcommand{\leq}{\leqslant}

\renewcommand{\geq}{\geqslant}

\newcommand{\dps}{\displaystyle}
\newcommand{\dt}{{\Delta t}}
\newcommand{\rme}{\mathrm{e}}
\newcommand{\sS}{\mathscr{S}}
\newcommand{\Er}{e_{\mathrm{min}}}
\newcommand{\Ef}{e_{\mathrm{max}}}

\newcommand{\Li}{\mathcal{K}}
\newcommand{\Id}{\mathrm{Id}}

\DeclareMathAlphabet{\mathpzc}{OT1}{pzc}{m}{it}
 \newcommand{\cL}{\mathcal{L}}
\newcommand{\cH}{\mathcal{H}}

\newcommand{\eps}{\varepsilon}
\renewcommand{\leq}{\leqslant}
\renewcommand{\geq}{\geqslant}
\newcommand{\cLham}{\mathcal{L}_\mathrm{ham}}
\newcommand{\cLovd}{\mathcal{L}_\mathrm{ovd}}
\newcommand{\cLFD}{\mathcal{L}_\mathrm{FD}}

\newcommand{\cM}{\mathcal{M}}

\begin{document}

\title{Langevin dynamics with general kinetic energies}			

\author{Gabriel Stoltz$^1$ and Zofia Trstanova$^2$ \\
 {\small $^{1}$ Universit\'e Paris-Est, CERMICS (ENPC), INRIA, F-77455 Marne-la-Vall\'ee, France} \\
{\small $^2$ University of Edinburgh, United Kingdom
}
}		

\maketitle

\abstract{  
We study Langevin dynamics with a kinetic energy different from the standard, quadratic one in order to accelerate the sampling of Boltzmann--Gibbs distributions. In particular, this kinetic energy can be non-globally Lipschitz, which raises issues for the stability of discretizations of the associated Langevin dynamics. We first prove the exponential convergence of the law of the continuous process to the Boltzmann--Gibbs measure by a hypocoercive approach, and characterize the asymptotic variance of empirical averages over trajectories. We next develop numerical schemes which are stable and of weak order two, by considering splitting strategies where the discretizations of the fluctuation/dissipation are corrected by a Metropolis procedure. We use the newly developped schemes for two applications: optimizing the shape of the kinetic energy for the so-called adaptively restrained Langevin dynamics (which considers perturbations of standard quadratic kinetic energies vanishing around the origin); and reducing the metastability of some toy models using non-globally Lipschitz kinetic energies.
}

%------------------------- intro --------------------------------------
\section{Introduction}

In statistical physics, the macroscopic information of interest for the systems under consideration can be inferred from averages over microscopic configurations distributed according to probability measures~$\mu$ characterizing the thermodynamic state of the system~\cite{Balian,Tuckerman}. Due to the high dimensionality of the system (which is proportional to the number of particles), these configurations are most often sampled using trajectories of stochastic differential equations or Markov chains ergodic for the probability measure~$\mu$; see for instance~\cite{LMbook,actaLelievre2016}. 

We focus here on a typical choice for~$\mu$, namely the Boltzmann--Gibbs measure, which describes a system at constant temperature. One popular stochastic process allowing to sample this measure is the Langevin dynamics. We denote the configuration of the system by $(q,p) \in \El$, where $q\in\mathcal{D}^d$ are the positions of the particles in the system (with $\mathcal{D} = \mathbb{R}$ or $\mathcal{D} = \mathbb{R}/\mathbb{Z}$ for systems with periodic boundary conditions), and $p \in \mathbb{R}^d$ the associated momenta. Therefore, $\mathcal{E} = \mathcal{D}^d \times \mathbb{R}^d$. For general separable Hamiltonian energies of the form $H(q,p) = V(q) + U(p)$, the Langevin dynamics reads
\begin{equation}
  \label{modified Langevin}
  \left\{
  \begin{aligned}
    dq_t & = \n U(p_t) \, dt, \\
    dp_t & = -\n V(q_t) \, dt - \gamma \n U(p_t) \, dt + \sqrt{\frac{2\gamma}{\beta}} \, dW_t,
  \end{aligned}
\right.
\end{equation}
where $dW_t$ is a standard $d$-dimensional Wiener process, $\beta > 0$ is proportional to the inverse temperature and $\gamma>0$ is the friction constant. The corresponding Boltzmann--Gibbs (or canonical) measure is
\begin{equation}
\label{eq: invariant measure modified Langevin}
\mu(dq \,dp)=Z^{-1}_{\mu}\mathrm{e}^{-\beta H(q,p)}\,dp\,dq, \qquad Z_{\mu}=\int_{\El}\mathrm{e}^{-\beta H(q,p)}\,dp\,dq.
\end{equation} 
Averages of an observable $\vp$ with respect to this distribution are approximated by ergodic means as 
\begin{equation}
\lim_{t \rightarrow\infty}\widehat{\vp}_t=\E_{\mu}(\vp)\quad \text{a.s.}, \qquad \widehat{\vp}_t:=\frac{1}{t}\int_0^t\vp(q_s,p_s) \, ds.
\label{eq: ergodic averages}
\end{equation}

In practice, the Langevin dynamics~\eqref{modified Langevin} cannot be analytically integrated. Its solution is therefore approximated with a numerical scheme. The numerical analysis of such discretization schemes is by now well-understood when $U$ is the standard quadratic kinetic energy. We refer for instance to~\cite{mattingly2002ergodicity,Kopec} for implicit schemes suited for dynamics in unbounded spaces, and to~\cite{bou2010long,bou2009pathwise,Matthews,abdulle2015long} for mathematical studies of the properties of splitting schemes. % (whose exponential convergence is proved for compact position spaces). 

One important limitation of the estimators $\widehat{\vp}_t$ in~\eqref{eq: ergodic averages} are their possibly large statistical errors. Under certain assumptions on~$U,V$ (see \emph{e.g.}~\cite{actaLelievre2016,trstanova2015errorAnalysisOfMLD} and references therein), it can be shown that a central limit theorem holds true, so that $\sqrt{t}[\widehat{\vp}_t-\E_{\mu}(\vp)]$ converges in law to a centered Gaussian distribution of variance $\sigma_\vp^2$. The asymptotic variance $\sigma_\vp^2$ may be large due to the metastability of the Langevin process, which occurs as soon as the probability measure~$\mu$ is multimodal (\textit{i.e.} it has modes of large probabilities separated by low-probability regions). Since the statistical error scales as $\sigma_\vp/\sqrt{t}$, there are three ways to decrease it \emph{at fixed computational time}: 
\begin{enumerate}[(i)]
\item decrease the value of the asymptotic variance~$\sigma_\vp$ by using variance reduction techniques (stratification, importance sampling, control variates, etc; see for instance the review in~\cite[Section~3.4]{actaLelievre2016});
\item increase the timestep $\dt$ in order to increase the simulated physical time $N_{\rm iter} \dt$ at fixed number of iterations. The most important limitations on~$\dt$ are related to the stability of the schemes under consideration;
\item decrease the computation cost of a single step in order to increase the number of iterations~$N_{\rm iter}$. 
\end{enumerate}

In this work, we consider the mathematical analysis and discretization of modified Langevin dynamics which improve the sampling of the Boltzmann--Gibbs distribution by introducing a kinetic energy function $U$ more general than the standard quadratic one. The stability of the numerical schemes is a major concern here, but we also discuss some importance sampling strategy in Section~\ref{sec:num_non_glob_Lip}. We have in fact two situations in mind:
\begin{enumerate}[(a)]
\item adaptively restrained Langevin dynamics~\cite{PRL-ARPS}, where the kinetic energy vanishes for small momenta, while it agrees with the standard kinetic energy for large momenta. The interest of this dynamics is that slow particles are frozen. The computational gain follows from the fact that the interactions between frozen particles need not be updated. A mathematical analysis of the asymptotic variance for this method is presented in~\cite{trstanova2015errorAnalysisOfMLD}, while the algorithmic speed-up, which allows to decrease the cost of a single iteration, is made precise in~\cite{trstanova2015speedUpARPS};
\item Langevin dynamics with kinetic energies growing more than quadratically at infinity, in an attempt to reduce metastability. Recall indeed that the marginal $\nu(dq) = Z_\nu^{-1} \mathrm{e}^{-\beta V(q)} \, dq$ of the canonical measure in the position variables is the crucial part to sample. The marginal distribution of~$\mu$ in the variable~$q$ is always~$\nu$, whatever the choice of the kinetic energy~$U$. The extra freedom provided by~$U$ can be used in order to reduce the metastability of the dynamics and hence the variance when the aim is to sample~$\nu$. 
\end{enumerate}

One of the main issues with the situations we consider is the stability of discretized schemes. Several works indicate that explicit discretizations of Langevin-type dynamics with non-globally Lipschitz force fields are often unstable (in the sense that the corresponding Markov chains do not admit invariant measures), see \emph{e.g.} \cite{mattingly2002ergodicity}. We face such situations here, even for compact position spaces, when $\nabla U$ is not globally Lipschitz. For adaptively restrained Langevin dynamics, the difficulties arise from the possibly abrupt transition from the region where the kinetic energy vanishes to the region where it coincides with the standard one. As for the stabilization of the Euler-Maruyama discretization of overdamped Langevin dynamics in~\cite{RT96}, we suggest to use a Metropolis acceptance/rejection step~\cite{MRRTT53,Hastings70} in order to ensure the stability of the methods under consideration. Such a stabilization leads to schemes which can be seen as one step Hybrid Monte Carlo (HMC)\footnote{Also called "Hamiltonian Monte-Carlo" in the statistics community.} algorithms~\cite{duane1987hybrid} with partial refreshment of the momenta, studied for instance in~\cite{bou2009pathwise} for the standard kinetic energy. Here, in order to obtain a weakly consistent method of fractional order~3/2 (it is not trivial to go beyond order~1 schemes when the fluctuation/dissipation cannot be analytically integrated), we rely on the Metropolis schemes developped for overdamped Langevin dynamics in~\cite{fathi2015improving}.

\medskip

This article is organized as follows. In Section~\ref{section mld}, we present the modified Langevin dynamics, give an exponential convergence result for the law of the process and make precise the asymptotic variance of empirical averages over a trajectory. We next discuss in Section~\ref{section ghmc} the discretization of the dynamics, and introduce in particular a generalized Hybrid Monte Carlo scheme of weak order~3/2. We then turn to numerical results relying on the stability properties of the Metropolized scheme. We first propose, for the adaptively restrained Langevin dynamics, a better kinetic energy function than the one originally suggested in~\cite{PRL-ARPS} (see Section~\ref{section AR-Langevin dynamics}); and finally demonstrate on a simple example how the choice of non-quadratic kinetic energies can dramatically improve the sampling efficiency (see Section~\ref{sec:num_non_glob_Lip}). The proofs of some technical results needed in the analysis of Section~\ref{section mld} are gathered in Appendix~\ref{sec:technical_results}. 

%--------------------------- rappels ---------------------
\section{Convergence of the modified Langevin dynamics}
\label{section mld}

We consider in all this work kinetic energies~$U$ and potentials~$V$ which satisfy the following conditions.

\begin{assumption}
The functions $U,V$ are smooth functions growing at most polynomially at infinity and such that
\[
\int_{\mathbb{R}^d} \mathrm{e}^{-\beta U} < +\infty, \qquad \int_{\mathcal{D}^d} \mathrm{e}^{-\beta V} < +\infty.
\]
\end{assumption}

We denote the generator of the dynamics~\eqref{modified Langevin} by 
\begin{equation}
\Lopgen = \Lop_{\rm Ham}+\gamma\Lop_{\rm{FD}}, \qquad \Lop_{\rm Ham} = \n \U \cdot \n_q - \n V\cdot \n_p, \qquad \Lop_{\rm{FD}} = -\n \U\cdot \n_p + \frac{1}{\beta}\de_p.
\label{eq: generator modified Langevin}
\end{equation}
A simple computation shows that~\eqref{modified Langevin} leaves the measure~\eqref{eq: invariant measure modified Langevin} invariant since, for all $C^\infty$ functions~$\vp$ with compact support,
\[
\int_{\El}\Lopgen\vp \ d\mu=0.
\]
We refer for instance to the review in~\cite{actaLelievre2016} for convergence results for the Langevin dynamics associated with the standard kinetic energy 
\begin{equation}
\label{eq:Ustd}
U_{\rm std}(p) = \frac12 p^T M^{-1}p,
\end{equation}
where $M$ is a positive mass matrix (typically a diagonal matrix, where the entries are the inverses of the masses of the particles in the system). These convergence results are stated either in terms of ergodic averages (Law of Large Numbers and Central Limit Theorem) or in terms of the law of the process at time~$t$. The aim of this section is to extend these results to more general kinetic energies. We start by providing an exponential convergence result for the law of the process in Section~\ref{sec:cv_law}, before studying in more detail the asymptotic variance of time averages in Section~\ref{sec:CLT}.

\subsection{Convergence of the law}
\label{sec:cv_law}

An extension of the hypocoercive approach of~\cite{DMS09,DMS15} allows to state exponential convergence results for the law of the process~\eqref{modified Langevin} in the Hilbert space~$L^2(\mu)$, for quite general kinetic energies, possibly non globally Lipschitz - in any case more general than the ones we considered in our previous work~\cite{trstanova2015errorAnalysisOfMLD}. This approach also turns out to be more quantitative since it provides upper bounds on the convergence rate which can be made explicit in terms of the friction~$\gamma$ and possibly other parameters of the dynamics (see for instance~\cite{RS17,IOS17} for similar results). Moreover, the result holds both for bounded and unbounded position spaces (contrarily to~\cite{trstanova2015errorAnalysisOfMLD} where the analysis is performed only for bounded position spaces).

In the following we consider all operators as defined on the Hilbert space $L^2(\mu)$ unless explicitly mentioned otherwise. The adjoint of a closed operator $T$ on $L^2(\mu)$ is denoted by $T^*$. The scalar product and norm on $L^2(\mu)$ are respectively denoted by $\langle \cdot, \cdot \rangle_{L^2(\mu)}$ and $\| \cdot \|_{L^2(\mu)}$. The norm of a bounded operator~$T$ on $L^2(\mu)$ is
\[
\|T\| = \sup_{g \in L^2(\mu) \backslash \{0\}}\frac{\|Tg\|_{L^2(\mu)}}{\|g\|_{L^2(\mu)}}.
\]
In this framework, the Fokker--Planck equation associated with~\eqref{modified Langevin} reads
\[
\partial_t f = \cL^* f,
\]
where $\psi(t) = f(t)\mu$ is the law of~\eqref{modified Langevin} at time~$t$, and 
\[
\cL^* = - \cLham + \gamma \cLFD.
\] 
Since \[ \int_\mathcal{E} \psi(0) = \int_\mathcal{E} f(0) \, d\mu = 1, \] it is expected that $f(t) = \rme^{t \cL^*} f(0)$ converges to the constant function~$\mathbf{1}$ as $t \to +\infty$. In order to state a precise convergence result, we need some conditions on both $U$ and~$V$, and on the marginal measures of $\mu$ in the position and momentum variables. These marginal probability measures are respectively
\begin{equation}
  \label{eq:marginal measures}
  \nu(dq) = Z^{-1}_\nu \rme^{-\beta V(q)} \, d q, 
  \qquad  
  \kappa(dp) = Z_\kappa^{-1} \rme^{-\beta U(p)} \, d p.
\end{equation}
Moreover, for any $\alpha = (\alpha_1,\dots,\alpha_d)$, we denote by $\partial_p^\alpha = \partial_{p_1}^{\alpha_1} \dots \partial_{p_d}^{\alpha_d}$ and $|\alpha| = \alpha_1 + \dots + \alpha_d$.

\begin{assumption}
\label{ass:UV}
The marginal measures $\nu$ and $\kappa$ satisfy Poincar\'e inequalities: There exist $K_\nu,K_\kappa > 0$ such that, for any $(\phi,\varphi) \in L^2(\nu) \times L^2(\kappa)$,  
\begin{equation}
\label{eq:Poincare}
\left\| \phi - \int_{\mathcal{D}^d} \phi \, d\nu \right\|_{L^2(\nu)} \leq \frac{1}{K_\nu} \| \nabla_q \phi  \|_{L^2(\nu)},
\qquad
\left\| \varphi - \int_{\R^d} \varphi \, d\kappa \right\|_{L^2(\kappa)} \leq \frac{1}{K_\kappa} \| \nabla_p \varphi  \|_{L^2(\kappa)}.
\end{equation}
We also assume the following 
\begin{enumerate}[(i)]
\item there exist $c_1 > 0$, $c_2 \in [0,1)$ and $c_3 > 0$ such that $V$ satisfies
  \begin{equation}
    \label{eq:regularization condition}
    \Delta V \leq c_1 + \frac {c_2} 2 | \nabla V |^2, \quad |\nabla^2 V | \leq c_3 \left( 1 + | \nabla V | \right);
  \end{equation}
\item the kinetic energy $U$ is such that $\partial_p^\alpha U$ belongs to $L^2(\kappa)$ for any $|\alpha| \leq 3$, and $(\partial^{\alpha} U) (\partial^{\alpha'} U)$ is in $L^2(\kappa)$ for $|\alpha| \leq 2$ and $|\alpha'| = 1$.
\end{enumerate}
\end{assumption}

Recall that there are various criteria ensuring that the Poincar\'e inequalities~\eqref{eq:Poincare} hold. One example is the following condition~\cite{BBCG08}: there exist $a_\nu,a_\kappa \in (0,1)$ such that
\[
\liminf_{|q| \to \infty} a_\nu \beta | \nabla V (q)|^2 - \Delta V(q) > 0,
\qquad
\liminf_{|p| \to \infty} a_\kappa \beta | \nabla U (p)|^2 - \Delta U(p) > 0.
\]
It is easy to check that the conditions in Assumption~\ref{ass:UV} are satisfied for $U$ and $V$ which asymptotically behave at infinity as $|q|^a$ and $|p|^b$, with $a,b>1$. Note also that the kinetic energy~$U$ is allowed to be constant on open sets, so that the generator~$\cL$ or its adjoint~$\cL^*$ are not necessarily hypoelliptic. Despite this possible lack of hypoellipticity, the following convergence result holds. In order to state it, we introduce the following subspace of $L^2(\mu)$: \[ L^2_1(\mu) = \left\{ f \in L^2(\mu) \, \left| \int_\mathcal{E} f \, d\mu = 1 \right. \right\}. \]

\begin{theorem}
\label{theorem DMS}
Suppose that Assumption~\ref{ass:UV} holds. Then, there exist $C,\lambda > 0$ such that, for any $\gamma \in (0,+\infty)$,
\[
\forall t \geq 0, \quad \forall f \in L^2_1(\mu), \qquad
\left\| \rme^{t \cL^*} f - \mathbf{1} \right\|_{L^2(\mu)} \leq C \rme^{-\lambda \min(\gamma,\gamma^{-1})t} \| f - \mathbf{1} \|_{L^2(\mu)}.
\]
\end{theorem}

Note that, as for standard kinetic energies, we find an upper bound of the form $\min(\gamma,\gamma^{-1})$ for the convergence rate. Let us mention that, unfortunately, we were not able to extract a meaningful dependence of~$\lambda$ on~$U$, see Remark~\ref{rmk:dep_U}. Such a result would be extremely useful since it would provide a theoretical guide for designing ``optimal'' kinetic energies.

Let us briefly sketch the proof of Theorem~\ref{theorem DMS}, which very closely follows the proof presented in~\cite[Appendix~A]{RS17} apart from some technical results requiring a dedicated treatment postponed to Appendix~\ref{sec:technical_results}. Introduce the projection $\Pi : L^2(\mu) \to L^2(\nu)$ defined as
\[
(\Pi g)(q) = \langle g(q,\cdot), \mathbf{1}\rangle_{L^2(\kappa)} = \int_{\R^d} g(q,p)\, \kappa(dp),
\]
as well as the operator 
\[
A = -\left(1 - \Pi \cLham^2 \Pi \right)^{-1} \Pi \cLham.
\]
In fact, $A$ is bounded with $\|A\| \leq 1/2$; see Lemma~\ref{lem:A} for further properties of this operator. We next consider the modified squared norm on~$L^2(\mu)$:
\begin{equation}
\label{eq:entropy_functional}
\cH(g) = \frac12 \| g \|_{L^2(\mu)}^2 + \eps \langle Ag, g\rangle_{L^2(\mu)},
\end{equation}
which is equivalent to the standard norm for $\eps \in (0,1)$; and denote by $\langle\langle\cdot,\cdot\rangle\rangle$ the scalar product associated with~$\cH$ by polarization. The key point is the following coercivity property, formulated for functions in~$\mathscr{C}$, the space of real valued $C^\infty$ functions with compact support (see the proof in Appendix~\ref{sec:technical_results}). In order to state it, we introduce the following subspace of $L^2(\mu)$: \[ L^2_0(\mu) = \left\{ f \in L^2(\mu) \, \left| \int_\mathcal{E} f \, d\mu = 0 \right. \right\}. \]

\begin{proposition}
  \label{prop:coercivity_scrD}
  There exists $\overline{\varepsilon} \in (0,1)$ and $\lambda>0$, such that, by considering $\varepsilon = \overline{\varepsilon} \min(\gamma,\gamma^{-1})$ in~\eqref{eq:entropy_functional},
  \begin{equation}
    \label{eq:coercivity double angle}
    \forall g \in \mathscr{C} \cap L^2_0(\mu), \qquad \langle \langle -\cL^* g,g \rangle \rangle \geq \widetilde{\lambda}_\gamma \|g\|^2,
  \end{equation}
  with $\widetilde{\lambda}_\gamma \geq \lambda \min(\gamma,\gamma^{-1})$.
\end{proposition}

This coercivity property and a Gronwall inequality then allow to conclude to the exponential convergence to~0 of $\cH[\rme^{t \cL^*} (f - \mathbf{1})]$, from which Theorem~\ref{theorem DMS} follows by the norm equivalence of $\sqrt{\cH}$ and $\| \cdot \|_{L^2(\mu)}$.

%----------------------------
\subsection{Asymptotic variance of empirical averages}
\label{sec:CLT}

We consider in this section the asymptotic behavior of the ergodic averages~\eqref{eq: ergodic averages}. The first result is an ergodicity property, which holds under the following assumption.

\begin{assumption}
  \label{ass:nabla2_U}
  The generator $\mathcal{L}$ is hypoelliptic. 
\end{assumption}

\begin{proposition}
Suppose that Assumption~\ref{ass:nabla2_U} holds. Then, for any bounded measurable function~$\varphi$, it holds
\[
\widehat{\varphi}_t \xrightarrow[t \to +\infty]{} \int_{\mathcal{E}} \varphi \, d\mu \qquad \mathrm{a.s.}
\]
\end{proposition}

The result is a consequence of~\cite{Kli87} since an invariant probability measure (namely~$\mu$) is known. A sufficient condition for $\mathcal{L}$ to be hypoelliptic is that the matrix $\nabla^2 U(p) \in \R^{d \times d}$ is definite positive for all $p \in \R^d$, see~\cite[Section~3.1]{trstanova2015errorAnalysisOfMLD}. Weaker conditions involving non-vanishing higher order derivatives could also be stated. Let us also mention that it is possible to remove the assumption that $\mathcal{L}$ is hypoelliptic as done in~\cite{trstanova2015errorAnalysisOfMLD} (where the derivatives of~$U$ vanish on a set of positive measure), but in this case only sufficiently small perturbations of standard quadratic kinetic energies can be considered, and the position space should be compact. 

Once the ergodicity of the dynamics is ensured, it is possible to characterize the asymptotic variance as a corollary of the convergence result provided by Theorem~\ref{theorem DMS}.

\begin{theorem}
  Suppose that Assumptions~\ref{ass:UV} and~\ref{ass:nabla2_U} hold. Then, for any $\varphi \in L^2(\mu)$,
  \[
  \lim_{t \to +\infty} t \mathbb{E}\left[ \left(\widehat{\varphi}_t - \int_{\mathcal{E}} \varphi \, d\mu\right)^2 \right] = \sigma^2_\varphi,
  \qquad 
  \sigma^2_\varphi = 2 \int_\mathcal{E} \left[-\cL^{-1}\left(\varphi - \int_{\mathcal{E}} \varphi \, d\mu\right)\right] \varphi \, d\mu, 
  \]
  where the expectation is over initial conditions $(q_0,p_0) \sim \mu$ and for all realizations of the Brownian motion in~\eqref{modified Langevin}.
\end{theorem}

The proof of this result is a simple consequence of a dominated convergence argument and the exponential convergence to~0 of the semigroup $\rme^{t \cL}$ on $L^2_0(\mu)$, which has the same operator norm as its adjoint $\rme^{t \cL^*}$. In particular, $\cL$ is invertible on $L^2_0(\mu)$; see~\cite[Section~3.1.2]{actaLelievre2016} for the complete argument. Let us also note that, using the results of~\cite{bhattacharya1982}, it is possible to state a Central Limit Theorem, even for initial conditions not distributed according to the canonical measure. 

%------------------------------------------------------------------------------------
\section{Discretization of the modified Langevin dynamics}
\label{section ghmc}

For a given timestep $\dt > 0$, numerical schemes approximate the solution $(q_{n\dt},p_{n\dt})$ of the Langevin dynamics~\eqref{modified Langevin} by $(q^n,p^n)$. The sequence $(q^n,p^n)_{n \geq 0}$ usually is a Markov chain. One appealing strategy to construct numerical schemes for Langevin dynamics is to resort to a splitting scheme between the Hamiltonian part of the dynamics (typically integrated with a Verlet scheme~\cite{Verlet}) and the fluctuation/dissipation dynamics on the momenta. The corresponding dynamics 
\begin{equation}
  \label{eq:elementary_FD}
  dp_t = -\gamma \n U(p_t) \, dt + \sqrt{\frac{2\gamma}{\beta}} \, dW_t,
\end{equation}
with generator $\gamma \Lop_{\rm FD}$, cannot be analytically integrated, except for very specific kinetic energies such as $U_{\rm std}$ defined in~\eqref{eq:Ustd}. A simple extension of the results of~\cite{Matthews} shows that splitting schemes (either Lie or Strang) based on a weakly second order consistent discretization of~\eqref{eq:elementary_FD} and a Verlet scheme for the Hamiltonian part are globally weakly consistent, of weak order~1 for Lie-based splittings and of weak order~2 for Strang based splittings. Moreover, in the case when the kinetic energy is a perturbation of the standard kinetic energy, in the sense that
\begin{equation}
  \label{eq:pert_U_std}
  \| \n U - \n U_{\rm std} \|_{L^\infty} < +\infty,
\end{equation}
it can be shown that the numerical schemes admit a unique invariant probability measure~$\mu_\dt$. Finally, it is possible to prove exponential convergence in some weighted $L^\infty$ spaces, with rates which are uniform in the timestep $\dt$ and depend only on the physically elapsed time. This allows also to state error estimates on the invariant measure~$\mu_\dt$ and on integrated correlation functions. Such results are obtained by adapting the proofs of the corresponding statements in~\cite{Matthews}, upon replacing $\n U_{\rm std}(p) = M^{-1}p$ with $\n U(p) = M^{-1}p + Z(p)$ where $Z$ is uniformly bounded (see~\cite{PhD}). 

On the other hand, when the condition~\eqref{eq:pert_U_std} is not satisfied, it may not be possible to prove the existence of a unique invariant measure for the splitting schemes. The main obstruction is that the Markov chain corresponding to the discretization of the elementary fluctuation/dissipation dynamics~\eqref{eq:elementary_FD} may itself be transient. This is the case for instance for non-globally Lipschitz force fields $\n U$ and a Euler-Maruyama discretization~\cite{RT96}. This observation motivates resorting to a Metropolis correction in order to ensure the existence of an invariant probability distribution.

We present in this section a generalized Hybrid Monte-Carlo (GHMC) scheme to discretize the Langevin dynamics with non-quadratic kinetic energies. For an introduction to HMC and some of its generalizations, we refer the reader to, for instance~\cite[Section 2.2.3]{lelievre2010free} and~\cite{BS18}. In essence, HMC is a Metropolis-Hastings method based on a proposal generated by the integration of the deterministic Hamiltonian dynamics. The proposal is then accepted or rejected according to a Metropolis rule. The rejection of the proposal occurs due to discretization errors. The efficiency of the method is therefore a trade-off between larger simulated physical times (which calls for larger timesteps) and not too large rejection rates (which places an upper limit on possible timesteps). 

We metropolize the Langevin dynamics with a general kinetic energy in two steps: first, we metropolize the Hamiltonian part as in the standard single-step HMC method (see Section~\ref{Metropolization of the Hamiltonian part}); in a second step, we add a weakly consistent discretization of the elementary fluctuation/dissipation stabilized by a Metropolis procedure (see Section~\ref{sec:disc_FD}). The complete algorithm is summarized in Section~\ref{Complete Generalized Hybrid Monte-Carlo scheme}. Let us already emphasize that the canonical measure is by construction an invariant measure for the numerical scheme. On the other hand, dynamical properties such as correlations in time are in general corrupted by the Metropolization procedure, which incurs stagnations due to rejected moves and may lead to large biases. This issue can be taken care of by constructing schemes with sufficiently high weak order, relying on standard weak type error estimates at finite times~\cite{MT04}.

In order to state rigorous results, we work with functions growing at most polynomially. More precisely, introducing the weight function $\Li_\alpha(q,p) = 1 + |q|^\alpha + |p|^\alpha$ for $\alpha \in \mathbb{N}$, we consider the following spaces of functions growing at most as $\Li_\alpha$ at infinity:
\[
L^\infty_{\Li_\alpha} = \left\{ f \textrm{ measurable}, \quad \left\|f \right\|_{L^\infty_{\Li_\alpha}} = \left\|\frac{f}{\Li_\alpha}\right\|_{L^\infty} < +\infty \right\}.
\]
In order to write more concise statements, we simply say that a family of functions $f_\dt$ grows at most polynomially in $(q,p)$ uniformly in~$\dt$ when there exist $K,\alpha,\dt^*>0$ such that 
\begin{equation}
\label{grows at most polynomially}
\sup_{0 < \dt \leq \dt^*} \left\|f_{\de t}\right\|_{L^\infty_{\Li_\alpha}} \leq K. 
\end{equation}
We finally define the vector space $\sS$ of smooth functions which, together with all their derivatives, grow at most polynomially.

\subsection{Metropolization of the Hamiltonian part}
\label{Metropolization of the Hamiltonian part}

Let us describe the one-step HMC method we use to discretize the Hamiltonian part of the dynamics:
\begin{equation}
  \label{eq:Hamiltonian_dyn}
  \left\{
  \begin{aligned}
    dq_t & = \n U(p_t) \, dt, \\
    dp_t & = -\n V(q_t) \, dt.
  \end{aligned}
\right.
\end{equation}
In order to ensure the reversibility of dynamics, we need to assume that the kinetic energy is symmetric: $U(p)=U(-p)$.
Starting from a configuration $(q^n, p^n)\in \El$, a new configuration $(\widetilde{q}^{n+1}, \widetilde{p}^{n+1}) = \Phi_\dt(q^n,p^n)\in \El$ is proposed using the Verlet scheme
\begin{equation}
\displaystyle
\left\{
\begin{aligned}
p^{n+1/2} &=  p^{n}-\n V(q^n)\frac{\de t}{2}, \\
\widetilde{q}^{n+1}&=q^n+\n U(p^{n+1/2})\de t, \\
\widetilde{p}^{n+1} &=  p^{n+1/2}-\n V(\widetilde{q}^{n+1})\frac{\de t}{2}. \\
\end{aligned}
\right.
\label{Verlet}
\end{equation}
The proposal is then accepted with probability
\begin{equation}
A_{\de t}^{\rm Ham}\left(q^n, p^n\right)=\min \left(1,\exp\left(-\beta\Big[H\left(\Phi_\dt(q^n,p^n)\right)- H\left(q^n, p^n \right)\Big]\right)\right).
\label{rejection rate hmc}
\end{equation}
If the proposal is rejected, a momentum reversal is performed and the next configuration is set to $(q^{n+1}, p^{n+1})=(q^n, -p^n )$ (see the discussion in~\cite[Section~2.2.3]{lelievre2010free} for a motivation of the momentum reversal). In summary, the new configuration is 
\begin{equation}
\begin{aligned}
\left(q^{n+1},p^{n+1}\right) &= \Psi^{\rm Ham}_\dt(q^n,p^n,\mathcal{U}^n) \\
& = \1_{\left\{\mathcal{U}^n\leq A^{\rm Ham}_{\de t} \left(q^n, p^n\right)\right\}}\Phi_{\de t}\left(q^n, p^n\right)+\1_{\left\{\mathcal{U}^n > A^{\rm Ham}_{\de t} \left(q^n, p^n\right)\right\}}\left(q^n, -p^n\right),
\end{aligned}
\label{eq: hmc}
\end{equation}
where $(\mathcal{U}^n)_{n \geq 0}$ is a sequence of independent and identically distributed (i.i.d.) random variables uniformly distributed in~$[0,1]$. A simple proof shows that the canonical measure~$\mu$ is invariant by the scheme~\eqref{eq: hmc}. The corresponding Markov chain is however of course not ergodic with respect to~$\mu$ since momenta are not resampled or randomly modified at this stage (this will be done by the discretization of the fluctuation/dissipation, see Section~\ref{Complete Generalized Hybrid Monte-Carlo scheme} for the complete GHMC scheme).

Without any discretization error (\textit{i.e.} if the Hamiltonian dynamics was exactly integrated, so that the energy would be constant), the proposal would always be accepted. Since the Verlet scheme is of order~2, we expect the energy difference $H\left(\Phi_\dt(q^n,p^n)\right)- H\left(q^n, p^n \right)$ to be of order~$\dt^3$. The following lemma makes this intuition rigorous and quantifies the rejection rate $1-A_\dt^{\rm Ham}$ in terms of the timestep $\dt$ and derivatives of the potential and kinetic energy functions.

\begin{lemma}
  \label{lemma rejection rate hmc}
  Assume that $U,V \in \sS$ { and $U$ is symmetric}. Then there exist $K,\dt^*,\alpha > 0$ such that the rejection rate of the one-step HMC scheme~\eqref{eq: hmc} admits the following expansion: for any $\dt \in (0,\dt^*]$,
  \begin{equation}
    0 \leq 1-A^{\rm Ham}_{\de t} = \de t^3\xi_++\de t^4r_{\de t}\,,
    \label{eq: rej rate HMC}
  \end{equation}
  with $\sup_{0 < \dt \leq \dt^*} \|r_{\de t}\|_{L^\infty_{\Li_\alpha}} \leq K$. Moreover, the leading order of the rejection rate is given by~$\xi_{+}:=\max\left(0,\xi\right)$ with
  \begin{equation}
    \label{eq: lead order rejection rate}
    \xi = -\Lop_{\rm Ham}H_2, 
      \qquad 
      H_2(q,p) = \frac{1}{12}\left[-\frac{1}{2}\n V(q)^T\n^2 U(p)\n V(q)+\n U(p)^T\n^2 V(q)\n U(p)\right].
  \end{equation}
\end{lemma}

As discussed in the introduction, the crucial part of the sampling usually is the sampling of the marginal~$\nu$ of the canonical measure~$\mu$ in the position variable. There is therefore some freedom in the choice of~$U$. The expression of the rejection rate~\eqref{eq: lead order rejection rate} suggests that~$U$ should be chosen such that derivatives of order up to~3 are not too large, in order for $\xi_+$ to be as small as possible. This remark is used in Section~\ref{section AR-Langevin dynamics} to improve the kinetic energy functions currently considered in adaptively restrained Langevin dynamics. 

\begin{proof}
The idea of the proof is that, according to results of backward analysis~\cite{hairer2006geometric}, the first order modified Hamiltonian $H + \dt^2 H_2$ should be preserved at order~$\dt^5$ over one timestep. The energy variation is therefore given, at dominant order, by $-\dt^2 [ H_2(\Phi_\dt(q,p))-H_2(q,p) ] \simeq -\dt^3 (\Lop_{\rm Ham} H_2)(q,p)$, which motivates the dominant term in the rejection rate. 

To identify $H_2$ and make the previous reasoning rigorous, we write the proposal~\eqref{Verlet} as
\[
\Phi_{\de t}\left(q,p\right)=
\left(
\begin{aligned}
&q+\n U\left( p-\n V(q)\frac{\de t}{2}\right)\de t \\
&p-\n V(q)\frac{\de t}{2}-\n V\left(q+\n U\left( p-\n V(q)\frac{\de t}{2}\right)\de t\right)\frac{\de t}{2} \\
\end{aligned}
\right),
\]
so that
\begin{equation}
\begin{aligned}
\Phi_{\de t}\left(q,p\right) & = \begin{pmatrix} q \\ p \end{pmatrix} + \dt \begin{pmatrix} \n U(p) \\ -\n V(q) \end{pmatrix} - \frac{\dt^2}{2} \begin{pmatrix} \n^2 U(p) \n V(q) \\ \n^2 V(q) \n U(p) \end{pmatrix} \\
& \ \ \ \ +  \frac{\dt^3}{4} \begin{pmatrix} \dps \frac12 D^3 U(p): \n V(q)^{\otimes 2 } \\[5pt] \n^2V(q) \n^2 U(p) \n V(q)- D^3V(q) :\n U(p)^{\otimes 2} \end{pmatrix} + \dt^4 R_\dt(q,p),
\label{taylor phi}
\end{aligned}
\end{equation}
where, for a smooth function~$A$, the vector $D^3 A(x) : v^{\otimes 2}$ has components $v^T \nabla^2 (\partial_{x_i} A)v$, and the remainder $R_\dt(q,p)$ grows at most polynomially in $(q,p)$, uniformly in~$\dt$ (this is easily seen by performing Taylor expansions with integral remainders). Denoting by $y = (q,p)^T$, we note that the Hamiltonian dynamics~\eqref{eq:Hamiltonian_dyn} can be reformulated as
\[
\dot{y} = F(y), \qquad F(y) = \begin{pmatrix} \n U(p) \\ -\n V(q) \end{pmatrix}.
\]
This implies that
\[
\ddot{y} = DF(y) F(y) = -\begin{pmatrix} \n^2 U(p) \n V(q) \\ \n^2 V(q) \n U(p) \end{pmatrix},
\]
and 
\[
\dddot{y} = \begin{pmatrix} D^3 U(p): \n V(q)^{\otimes 2 } - \n^2 U(p) \n^2 V(q) \n U(p) \\ -D^3 V(p): \n U(p)^{\otimes 2 } + \n^2 V(q) \n^2 U(p) \n V(q)\end{pmatrix}.
\]
Therefore, denoting by $\phi_t$ the flow of the Hamiltonian dynamics~\eqref{eq:Hamiltonian_dyn}, it holds
\begin{equation}
  \label{eq:error_expansion_Phi}
  \Phi_\dt(q,p) = \phi_\dt(q,p) + \dt^3 G(q,p) + \dt^4 \widetilde{R}_\dt(q,p),
\end{equation}
where
\[
G(q,p) = \frac{1}{12}\left(
\begin{aligned}
&-\frac{1}{2}D^3 U(p):\n V(q)^{\otimes 2}+2\n^2 U(p)\n^2 V(q)\n U(p) \\
&-D^3V(q):\n U(p)^{\otimes 2}+\n^2 V(q)\n^2 U(p)\n V(q)
\end{aligned}
\right)\,,
\]
with a remainder $\widetilde{R}_\dt(q,p)$ growing at most polynomially in $(q,p)$ uniformly in~$\dt$. A simple computation shows that 
\[
G=\begin{pmatrix} \n_p H_2(q,p) \\ -\n_q H_2(q,p) \end{pmatrix}, 
\]
with $H_2$ defined in~\eqref{eq: lead order rejection rate}. Note that for the standard kinetic energy $U_{\rm std}$, this expression reduces to the one derived in~\cite{hairer2003geometric,sweet2009separable}.

From the error estimate~\eqref{eq:error_expansion_Phi}, we compute
\[
\begin{aligned}
H(\Phi_\dt(q,p)) - H(q,p) & = H(\phi_\dt(q,p)) - H(q,p) + \dt^3 G(q,p) \n H(q,p) + \dt^4 \widehat{R}_\dt(q,p) \\
& = -\dt^3 \Lop_{\rm Ham} H_2(q,p) + \dt^4 \widehat{R}_\dt(q,p),  
\end{aligned}
\]
where the remainder $\widehat{R}_\dt(q,p)$ grows at most polynomially in $(q,p)$ uniformly in~$\dt$. This allows to identify $\xi = -\Lop_{\rm Ham} H_2$ as the leading order term of the energy variation over one step. In order to compute the expected rejection rate, we rely on the inequality 
\[
x_+-\frac{x_+^2}{2}\leq 1-\min \left(1,{\rm e}^{-x}\right)\leq x_{+}, \quad x_{+}=\max (0,x).
\]
This implies that 
\begin{equation}
0 \leq 1 - A_{\Delta t}^{\rm Ham}\left(q^{n},p^n\right) = \de t^3\xi_{+}\left(q^{n},p^n\right) + \dt^4 \mathscr{R}_{\de t}(q^n,p^n)\,,
\label{eq: rejection rate proof}
\end{equation}
where the remainder $\mathscr{R}_{\de t}$ grows at most polynomially in $(q,p)$ uniformly in~$\dt$, which concludes the proof.
\end{proof}

As a corollary of the estimates~\eqref{eq: rej rate HMC} on the rejection rate and the consistency result~\eqref{eq:error_expansion_Phi} for the scheme without rejections, we can obtain weak-type expansions  of order 2 for the evolution operator
\[
P_\dt^{\rm Ham}\vp(q,p) = \mathbb{E}_\mathcal{U}\left[\vp\left(\Psi^{\rm Ham}_\dt(q,p,\mathcal{U})\right)\right].
\]

\begin{corollary} 
Assume that $U,V \in \mathscr{S}$  {  and $U$ is symmetric}. Then, for any $\vp \in \mathscr{S}$, there exist $\de t^*, K, \alpha>0$ such that
\[
P_{\de t}^{\rm Ham}\vp  = \vp +\de t \Lop_{\rm Ham}\vp +\frac{\de t^2}{2}\Lop_{\rm Ham}^2\vp + \de t^3R_{\de t}^{\rm Ham}\vp,
\]
where $\sup_{0<\de t \leq \de t^*}\left\| R_{\de t}^{\rm Ham} \vp\right\|_{L^{\infty}_{\K_{\alpha}}}\leq K$.
\end{corollary}

\begin{proof}
We write the generator of the Hamiltonian part as
\[
P_\dt^{\rm Ham}\vp(q,p) = \vp \left(\Phi_{\de t}(q,p)\right)+\left(1-A_{\de t}^{\rm Ham}(q,p)\right)\Big(\vp (q,-p)-\vp(\Phi_{\de t}(q,p))\Big).
\]
Since $A_\dt^{\rm Ham}(q,p) \in [0,1]$ and $\Phi_\dt(q,p)$ grows at most polynomially in $(q,p)$ uniformly in~$\dt$, a direct inspection of the latter expression shows that the operator $P_\dt^{\rm Ham}$ maps functions growing at most polynomially into functions growing at most polynomially: for any $\alpha \in \mathbb{N}$, there exist $\alpha' \in \mathbb{N}$ and $C_\alpha > 0$ such that
\begin{equation}
\label{eq:stab_L_infty_alpha_Ham}
\forall f \in L^\infty_{\Li_\alpha}, \qquad \left\|P_{\de t}^{\rm Ham}f \right\|_{L^\infty_{\Li_{\alpha}}} \leq C_\alpha \|f \|_{L^\infty_{\Li_{\alpha'}}}.
\end{equation}
In order to understand the behavior of the evolution operator for small~$\dt$, we first note that, for instance by the techniques reviewed in~\cite[Section~4.3]{Matthews}, it can be shown that, for any $\vp \in \sS$, 
\[
\vp \left(\Phi_{\de t}(q,p)\right) = \left(\vp + \dt \Lop_{\rm Ham}\vp + \frac{\dt^2}{2} \Lop^2_{\rm Ham}\vp + \dt^3 R_\dt^{\rm Verlet} \vp\right)(q,p), 
\]
where $R_\dt^{\rm Verlet} \vp$ grows at most polynomially in~$(q,p)$ uniformly in~$\dt$. Therefore, by~\eqref{eq: rejection rate proof}, 
\begin{equation}
  \label{eq:error_P_dt_Ham}
  P_\dt^{\rm Ham}\vp = \vp + \dt \Lop_{\rm Ham}\vp + \frac{\dt^2}{2} \Lop^2_{\rm Ham}\vp + \dt^3 R_{\dt}^{\rm Ham}\vp,
\end{equation}
where the remainder
\[
R_\dt^{\rm Ham}\vp(q,p)=\frac{1-A_{\de t}^{\rm Ham}(q,p)}{\dt^3}\Big(\vp (q,-p)-\vp(\Phi_{\de t}(q,p))\Big) + R_\dt^{\rm Verlet} \vp(q,p).
\] 
grows at most polynomially in $(q,p)$ uniformly in~$\dt$. 
\end{proof}

\subsection{Discretization of the fluctuation-dissipation}
\label{sec:disc_FD} 

In order to construct a GHMC scheme for \eqref{modified Langevin}, we need to generate momenta distributed according to 
\begin{equation}
  \label{eq:def_kappa}
  \kappa(dp) = Z_\kappa^{-1} {\rm e}^{-\beta U(p)} \, dp,
\end{equation}
which are then used as initial conditions in the Hamiltonian part of the scheme. This can be achieved through a discretization of the fluctuation-dissipation, corrected by a Metropolis procedure. 

We use here a scheme proposed in~\cite{fathi2015improving} for the elementary dynamics~\eqref{eq:elementary_FD}. The proposal function is given by
\begin{equation}
\widetilde{p}^{n+1}=\Phi_{\de t}^{\rm FD}(p^n, G^n) = p^n - \gamma \n U\left(p^n +\frac{1}{2}\sqrt{\frac{2\gamma\de t}{\beta}}G^n\right)\de t +\sqrt{\frac{2\gamma\de t}{\beta}}G^n\,,
\label{eq: FD verlet compact}
\end{equation}
where $(G^n)_{n \geq 0}$ is a sequence of i.i.d. standard $d$-dimensional Gaussian random variables. It seems that the computation of the probability density, to go from a given momentum $p$ to a new one $p'$, is difficult since $\Phi_{\de t}^{\rm FD}(p, G)$ depends nonlinearly on~$G$. It turns out however that the proposal~\eqref{eq: FD verlet compact} can itself be interpreted as the output of some one-step HMC scheme, starting from a random conjugate variable $R^n:=G^n/\sqrt{\beta}\in \R^d$ and for an effective timestep $h=\sqrt{2\gamma \de t}$: 
\begin{equation}
\left\{
\begin{aligned}
\ds p^{n+1/2} &= p^n +R^n\frac{h}{2}, \\
R^{n+1}&=R^n-\n U(p^{n+1/2})h, \\
\widetilde{p}^{n+1} &= p^{n+1/2} +R^{n+1}\frac{h}{2}. \\
\end{aligned}
\right.
\label{eq verlet gaussians}
\end{equation}
The Hamiltonian dynamics which is discretized by this scheme is the one associated with the energy
\[
E(p,R) = U(p) + \frac12 R^2.
\]
Therefore, the acceptance rule for the proposal~\eqref{eq: FD verlet compact} is
\[
A_\dt^{\rm FD}(p^n,G^n) = \min \left(1,\exp\left(-\beta\Big[ E\left(\widetilde{p}^{n+1},R^{n+1}\right) - E(p^n,R^n)\Big]\right)\right).
\]
In summary, the new momentum is therefore given by
\begin{equation}
\label{eq:HMC_like_proposal_for_Metropolis}
p^{n+1} = \Psi_\dt^{\rm FD}(p^n,G^n,\mathcal{U}^n) = p^n+\1_{\left\{\mathcal{U}^n\leq  A_{\de t}^{\rm FD}\left(p^n,G^n\right)\right\}}\left(\Phi_{\de t}^{\rm FD}(p^n, G^n)-p^n\right). 
\end{equation}

\begin{remark}
\label{remark dimension}
  Note that the efficiency of the Metropolization procedure of the fluctuation-dissipation does not degrade as the dimension increases when the kinetic energy is a sum of individual contributions, namely $d=ND$ (with usually $D \in \{1,2,3\}$) and 
  \[
  U(p)=\sum_{i=1}^N u(p_i),
  \]
  where $p_i\in \mathbb{R}^D$. Indeed, in this case, the dynamics in each component~$p_i$ are independent and can therefore be Metropolized independently one of another. More precisely, we consider in this case individual acceptance rates
  \[
  A_\dt^{{\rm FD,i}}(p_i^n,G_i^n) = \min \left(1,\exp\left(-\beta\Big[ E_i\left(\widetilde{p}_i^{n+1},R_i^{n+1}\right) - E_i(p_i^n,R_i^n)\Big]\right)\right),
  \]
  where $R^{n+1} = (R_1^{n+1},\dots,R_N^{n+1})$ with $R_i:=G_i^n/\sqrt{\beta}\in \R^D$ and $(G_i^n)$ is a sequence of i.i.d $D$-dimensional Gaussian random variables; and, for any $i\in \left\{ 1,\ldots, N\right\}$, the individual energies are
  \[
  E_i(p_i,R_i) = u(p_i) + \frac12 R_i^2.
  \]
  It is then possible to choose a timestep $\Delta t$ such that the average individual acceptance rates are, say, of order~1/2 (or any value in~$(0,1)$). In particular, the acceptance or rejection of the proposed move of one degree of freedom has no impact on the other ones. Note that the timestep therefore does not depend on the number of particles~$N$, in contrast with Metropolis dynamics which perform a global acceptance/rejection where the proposed moves for all degrees of freedom are either accepted or rejected at the same time. In the latter situation, the timestep should be chosen as some inverse fractional power of the number of degrees of freedom~$d=ND$ in order for the acceptance/rejection rate not to degrade as~$d$ increases (see for instance~\cite{RS98}). 
\end{remark}

In \cite{fathi2015improving}, the properties of the scheme~\eqref{eq: FD verlet compact} were studied for compact spaces. It is however possible to adapt some of the results obtained in this work for dynamics in unbounded spaces, upon introducing additional assumptions on the kinetic energy function. 

\begin{assumption}
\label{ass:moments}
The marginal measure~$\kappa$ defined in~\eqref{eq:def_kappa} admits moments of all orders: for all $k \in \mathbb{N}$, there exists $M_k < +\infty$ such that 
\[
\int_\mathcal{E} |p|^k \, \kappa(dp) \leq M_k.
\]
\end{assumption}

We can then state the following weak type expansion for the evolution operator 
\[
P_{\de t}^{\rm FD}\vp(p) = \mathbb{E}_{\mathcal{U},G}\left[\vp\left(\Psi_\dt^{\rm FD}(p,G,\mathcal{U})\right)\right].
\]

\begin{lemma}
\label{lem:weak_FD}
Suppose that $U \in \mathscr{S}$ and that Assumption~\ref{ass:moments} holds. Then, for any $\vp \in \sS$, 
\begin{equation}
  \label{eq:error_P_dt_FD}
  P_{\de t}^{\rm FD}\vp = \vp + \de t \Lop_{\rm FD} \vp +\frac{\de t^2}{2}\Lop_{\rm FD}^2\vp + \de t^{5/2} R^{\rm FD}_{\de t}\vp,
\end{equation}
where the remainder $R_\dt^{\rm FD}\vp$ grows at most polynomially in~$(q,p)$ uniformly in~$\dt$. Moreover, the rejection rate is of order~$\dt^{3/2}$: there exist a function $\zeta_+ \in \mathscr{S}$ as well as $K,\dt^*>0$ and $\alpha \in \mathbb{N}$ such that 
\[
0 \leq 1-\mathbb{E}_G\left[A^{\rm FD}_{\de t}(p,G)\right] = \de t^{3/2} \zeta_+(p) + \de t^2 r_{\de t}(p)\,,
\]
with $\sup_{0 < \dt \leq \dt^*} \|r_{\de t}\|_{L^\infty_{\Li_\alpha}} \leq K$. Finally, $P_{\de t}^{\rm FD}$ maps functions growing at most polynomially into functions growing at most polynomially: for any $\alpha \in \mathbb{N}$, there exist $\alpha' \in \mathbb{N}$ and $C_\alpha > 0$ such that
\begin{equation}
\label{eq:stab_L_infty_alpha_FD}
\forall f \in L^\infty_{\Li_\alpha}, \qquad \left\|P_{\de t}^{\rm FD}f \right\|_{L^\infty_{\Li_{\alpha}}} \leq C_\alpha \|f \|_{L^\infty_{\Li_{\alpha'}}}.
\end{equation}
\end{lemma}

An important comment at this stage is that the leading order remainder in~\eqref{eq:error_P_dt_FD} involves a fractional power of the timestep, whereas it would be of order~$\dt^3$ for standard discretization schemes of weak order~2. This is typical of Metropolis-like dynamics, as already noted in~\cite{FHS14,fathi2015improving} for instance. 

The proof of the first two properties in Lemma~\ref{lem:weak_FD} is a direct extension of~\cite[Lemma~3]{fathi2015improving} and its proof, and is therefore omitted. In fact, the scaling of the rejection rate could be obtained by a result similar to Lemma~\ref{lemma rejection rate hmc} for the effective timestep~$h = \sqrt{2\gamma\dt}$ in view of the reformulation~\eqref{eq verlet gaussians}. For the last property, we rely on the equality  
\[
P_{\de t}^{\rm FD}\vp(p) = \mathbb{E}_{G}\left[A^{\rm FD}_\dt(p,G) \vp\left(\Phi_\dt^{\rm FD}(p,G)\right)\right] + \left(1-\mathbb{E}_{G}\left[A^{\rm FD}_\dt(p,G)\right]\right)\vp(p),
\]
as well as on the fact that $\Phi_\dt^{\rm FD}(p,G)$ grows at most polynomially in~$(p,G)$ uniformly in~$\dt$.

\subsection{Complete Generalized Hybrid Monte-Carlo scheme}
\label{Complete Generalized Hybrid Monte-Carlo scheme}

The complete scheme for the metropolized Langevin dynamics with general kinetic energy is obtained by concatenating the updates~\eqref{eq: hmc} and~\eqref{eq:HMC_like_proposal_for_Metropolis}. Depending on whether Lie or Strang splittings are considered, and also on the order in which the operations are performed, several schemes can be considered. For instance, the scheme characterized by the evolution operator $P_\dt^{\rm GHMC} = P_\dt^{\rm FD}P_\dt^{\rm Ham}$ corresponds to first updating the momenta with~\eqref{eq:HMC_like_proposal_for_Metropolis}, and then updating both positions and momenta according to~\eqref{eq: hmc}.

All such splitting schemes preserve the invariant measure $\mu$ by construction. They are also all of weak order at least~1. A higher order weak accuracy can however be obtained for Strang splittings, as made precise in the following proposition.

\begin{proposition}
  \label{lemma evolution operator ghmc}
  Assume that $U,V \in \sS$, that Assumption~\ref{ass:moments} holds  {  and $U$ is symmetric}. Consider $P_{\de t}^{\rm GHMC} = P_{\dt/2}^{\rm FD}P_\dt^{\rm Ham}P_{\dt/2}^{\rm FD}$ or $P_{\de t}^{\rm GHMC} = P_{\dt/2}^{\rm Ham}P_\dt^{\rm FD}P_{\dt/2}^{\rm Ham}$. Then, for any $\vp \in \sS$, there exist $\dt^*,K,\alpha > 0$ such that
  \begin{equation}
    \label{eq: expansion Metropolis complete}
    P_{\de t}^{\rm GHMC}\vp = \vp + \dt \Lop \vp + \frac{\de t^2}{2}\Lop^2\vp +\de t^{5/2}r_{\de t, \vp},
  \end{equation}
  where $\sup_{0 < \dt \leq \dt^*} \|r_{\de t, \vp}\|_{L^\infty_{\Li_\alpha}} \leq K$. 
\end{proposition}

As in Lemma~\ref{lem:weak_FD}, we see the appearance of fractional powers of the timestep in the remainder. 

\begin{proof}
This result is a direct consequence of the estimates~\eqref{eq:error_P_dt_Ham} and~\eqref{eq:error_P_dt_FD}. We however sketch the proof for completeness. Fix $\vp \in \sS$. In view of~\eqref{eq:error_P_dt_FD}, 
\[
P_{\dt/2}^{\rm FD}P_\dt^{\rm Ham}P_{\dt/2}^{\rm FD}\varphi = P_{\dt/2}^{\rm FD}P_\dt^{\rm Ham}\widetilde{\vp} + \dt^{5/2} P_{\dt/2}^{\rm FD}P_\dt^{\rm Ham}R^{\rm FD}_{\de t}\vp, 
\]
where
\[
\widetilde{\vp} = \left(\Id + \frac{\de t}{2} \Lop_{\rm FD}  +\frac{\de t^2}{8}\Lop_{\rm FD}^2\right)\vp \in \sS.
\]
The remainder $P_{\dt/2}^{\rm FD}P_\dt^{\rm Ham}R^{\rm FD}_{\de t}\vp$ grows at most polynomially in~$(q,p)$ uniformly in~$\dt$ by~\eqref{eq:stab_L_infty_alpha_Ham} and~\eqref{eq:stab_L_infty_alpha_FD}. We next use~\eqref{eq:error_P_dt_Ham} to write 
\[
P_{\dt/2}^{\rm FD}P_\dt^{\rm Ham}\widetilde{\vp} = P_{\dt/2}^{\rm FD}\widehat{\vp} + \dt^3 P_{\dt/2}^{\rm FD}
R_\dt^{\rm Ham}\widetilde{\vp}, 
\]
where
\[
\widehat{\vp} = \left(\Id + \de t \Lop_{\rm Ham}  +\frac{\de t^2}{2}\Lop_{\rm Ham}^2\right)\left(\Id + \frac{\de t}{2} \Lop_{\rm FD}  +\frac{\de t^2}{8}\Lop_{\rm FD}^2\right)\vp \in \sS.
\]
The remainder $P_{\dt/2}^{\rm FD}R_\dt^{\rm Ham}\widetilde{\vp}$ grows at most polynomially in~$(q,p)$ uniformly in~$\dt$ by~\eqref{eq:stab_L_infty_alpha_FD}. By applying again~\eqref{eq:error_P_dt_FD}, we finally obtain that
\[
\begin{aligned}
& P_{\dt/2}^{\rm FD}P_\dt^{\rm Ham}P_{\dt/2}^{\rm FD}\varphi = \dt^{5/2} \mathcal{R}_{\dt,\varphi} \\
& + \left(\Id + \frac{\de t}{2} \Lop_{\rm FD}  +\frac{\de t^2}{8}\Lop_{\rm FD}^2\right)\left(\Id + \de t \Lop_{\rm Ham}  +\frac{\de t^2}{2}\Lop_{\rm Ham}^2\right)\left(\Id + \frac{\de t}{2} \Lop_{\rm FD}  +\frac{\de t^2}{8}\Lop_{\rm FD}^2\right)\vp,
\end{aligned}
\]
where the remainder $\mathcal{R}_{\dt,\varphi}$ grows at most polynomially in~$(q,p)$ uniformly in~$\dt$. The conclusion follows by expanding the last term on the right-hand side, grouping together terms of order~$\dt$ and~$\dt^2$, and gathering the higher order terms in the remainder.
\end{proof}

As a corollary of the weak error expansion~\eqref{eq: expansion Metropolis complete}, finite time weak type error estimates can be obtained by standard techniques under some technical conditions on~$U,V$; see~\cite[Chapter~2]{MT04} for a general presentation of these techniques, and for instance~\cite{LMT14} for an application to Langevin dynamics. Under these conditions, for a given sufficiently smooth observable~$\varphi$ and a fixed time $T > 0$, there is a constant $C_{T,\varphi}$ such that 
\begin{equation}
  \label{eq:finite_time_weak_error_GHMC}
  \sup_{0 \leq n \leq T/\dt} \Big| \mathbb{E}\left[ \varphi(q^n) \right] - \mathbb{E}\left[ \varphi(q_{n\dt}) \right] \Big| \leq C_{T,\varphi} \dt^{3/2}. 
\end{equation}
If the fluctuation/dissipation was integrated with a standard Metropolis Adjusted Langevin Algorithm (MALA)~\cite{RDF78,RT96}, i.e. the proposal~\eqref{eq: FD verlet compact} was replaced by $\widetilde{p}^{n+1}  = p^n - \gamma \nabla U(p^n) \dt + \sqrt{2\gamma \beta^{-1}\dt}\, G^n$, then an error estimate similar to~\eqref{eq:finite_time_weak_error_GHMC} would hold, but with a larger term~$\dt$ instead of~$\dt^{3/2}$ on the right-hand side. Such error estimates are illustrated in Section~\ref{sec:weakErrorNumerics}. 

On the other hand, it is much more difficult to prove error estimates for infinitely long times, such as time integrated correlation functions (Green--Kubo type formulas). One framework to this end is provided in~\cite{Matthews} and relies on an exponential convergence of $(P_\dt^{\rm GHMC})^n \vp$ towards $\mathbb{E}_\mu(\vp)$, uniformly in the spaces $L^\infty_{\Li_\alpha}$, and, most importantly, with a rate depending on the physical time $n\dt$, uniformly in~$\dt$. A typical way to obtain such estimates is to establish a Lyapunov condition for the functions $\Li_\alpha$ and a minorization condition on a compact space, in order to apply the results from~\cite{MeynTweedie,HairerMattingly-Yet}. Although we were able to prove a minorization condition in the case when $U - U_{\rm std}$ is bounded and the position space $\mathcal{D}$ is compact (see~\cite{PhD}), we were not able to establish a Lyapunov condition. The problem is that, even for compact position spaces and standard, quadratic kinetic energies, the rejection rate of the fluctuation/dissipation part of the scheme degenerates as $|p|\to +\infty$. Such difficulties were already encountered in the study of Metropolized Langevin-type algorithms on unbounded spaces, where the problem was taken care of by an appropriate truncation of the accessible space~\cite{BH13}.

%-----------------------------
\section{Applications}
\label{section applications}

We present in this section simulation results for the Langevin dynamics~\eqref{modified Langevin}. We consider two applications. The first one is the optimization the shape of the kinetic energy in the Adaptively Restrained Langevin dynamics (Section~\ref{section AR-Langevin dynamics}). The second one potentially has a much more important impact since we show that an appropriate choice of the kinetic energy can alleviate metastable features of Langevin dynamics and hence improve the sampling of probability measures (see Section~\ref{sec:num_non_glob_Lip}). We also illustrate on this second example the weak error estimates~\eqref{eq:finite_time_weak_error_GHMC} which show that average dynamical properties are well reproduced with the scheme we use.

\subsection{Adaptively restrained Langevin dynamics}
\label{section AR-Langevin dynamics}

The Adaptively Restrained Particle Simulation method was proposed in \cite{PRL-ARPS} in order to reduce the computational complexity of the forces update. The aim of this section is to devise better kinetic energy functions for the adaptively restrained (AR) Langevin dynamics, allowing for larger timesteps in the simulations. We start by recalling the kinetic energy function used in the original AR Langevin dynamics~\cite{PRL-ARPS} in Section~\ref{sec:spline}, where we also propose an alternative kinetic energy function. The relevance of this alternative energy function is studied in Section~\ref{sec:efficiency_U}, where we use the rejection rates of the GHMC algorithm to quantify the stability of the schemes under consideration. In essence, we fix an admissible rejection rate, and find the largest timestep for which the rejection rate is lower or equal to this tolerance. We therefore see the rejection rate as a measure of the stability, understood in this section as taking timesteps as large as possible while maintaining an appropriate consistency in the energy variation.

\subsubsection{Kinetic energy functions for AR Langevin}
\label{sec:spline}

In AR Langevin, the standard kinetic energy is replaced by a kinetic energy which vanishes for small values of momenta and matches the standard kinetic energy for sufficiently large values of momenta. The transition between these two regions is made in the original model~\cite{PRL-ARPS} by an interpolation spline $s_{\rm org}$ which ensures the regularity of the transition on the kinetic energy itself. More precisely, introducing two energy parameters $0 < \Er < \Ef$,
 \begin{equation}
\begin{aligned}
\displaystyle \UO (p)=\sum_{i=1}^N u(p_i)
%\label{eq: Ui}
 \quad\text{ where }\quad
\ds
u(p_i)=
\left\{
\begin{aligned}
  0 & \qquad \text{ for} \quad\frac{p_i^2}{2m_i}\leq \Er,\\
  s_{\rm org}\left(\frac{p_i^2}{2m_i}\right) & \qquad  \text{ for}\quad \frac{p_i^2}{2m_i}\in\left[\Er, \Ef\right],\\
  \frac{p_i^2}{2m_i}& \qquad  \text{ for} \quad \frac{p_i^2}{2m_i}\geq \Ef. 
\end{aligned}
\right.
\end{aligned}
\label{def arps kinetic energy}
 \end{equation}
The function $s_{\rm org}$ is such that $x \mapsto s_{\rm org}(x) \, \1_{x \in [\Er,\Ef]} + x \, \1_{x > \Ef}$ is $C^2(\R_+)$. The original AR Langevin kinetic energy was motivated by some physical interpretation in terms of momentum-dependent masses. One unpleasant feature of the definition~\eqref{def arps kinetic energy} is that the derivatives $\nabla U$ which appear in the dynamics~\eqref{modified Langevin} are typically large at the transition points (see Figure~\ref{fig:nablaU}). Since the dynamics is determined by $\nabla U$, a more satisfactory approach seems to interpolate the kinetic force $\nabla U$ between~0 in the region of small momenta and $M^{-1}p$ in the region of large momenta. We introduce to this end a second spline function $s_{\rm new}$ and define, for two velocity parameters $0 < \er < \ef$,
\begin{equation}
  \begin{aligned}
    \displaystyle \UN(p)=\sum_{i=1}^d u(p_i)
    %\label{eq: Ui}
    \quad\text{ where }\quad
    \ds
    u(p_i)=
    \left\{
    \begin{aligned}
      S_{\er\ef} &  \text{ for} \quad \frac{\left|{p_{i}}\right|}{m_i}\leq \er,\\% \leftarrow\text{ restraining parameter} \\
      s_{\rm new}\left(p_i\right) &  \text{ for} \quad\frac{\left|{p_{i}}\right|}{m_i}\in\left[\er, \ef\right],\\% \leftarrow\text{ restraining
      \frac{p_i^2}{2m_i}&  \text{ for} \quad \frac{\left|{p_{i}}\right|}{m_i}\geq \ef
    \end{aligned}
    \right.
  \end{aligned}
  \label{def new arps kinetic energy}
\end{equation}
where $S_{\er\ef}$ is a constant ensuring the continuity of the kinetic energy. Figures~\ref{fig:modifiedHamiltonian_parametrization} and~\ref{fig:nablaU} compare the original and new kinetic energies and their derivatives. %Figure~\ref{fig:modifiedHamiltonian_parametrization} represents the alternative kinetic energy~\eqref{def new arps kinetic energy} as a function of the momenta for various choices of the parameters. Figure~\ref{fig:nablaU} compares the derivatives of the original and new kinetic energies.
 Note that the alternative kinetic energy~\eqref{def new arps kinetic energy} leads to a smaller maximal value of the kinetic force $\n U$ than the original AR kinetic energy~\eqref{def arps kinetic energy}. This is also true for higher order derivatives of~$U$. 

It is difficult to directly compare the canonical distributions of momenta associated with $\UO$ and~$\UN$. For instance, it is not possible in general to ensure that these two distributions coincide for small and large momenta, because of the normalization constant in the probability distribution. In the sequel, we consider $\Er = m_i \er^2/2$ and $\Ef = m_i \ef^2/2$ for the $i$th particle, in order to have a constant kinetic energy (resp. a standard kinetic energy) in the same energy intervals.

\begin{figure}[t]
  \centering
    \begin{subfigure}[b]{0.45\textwidth}
      \includegraphics[width=\textwidth]{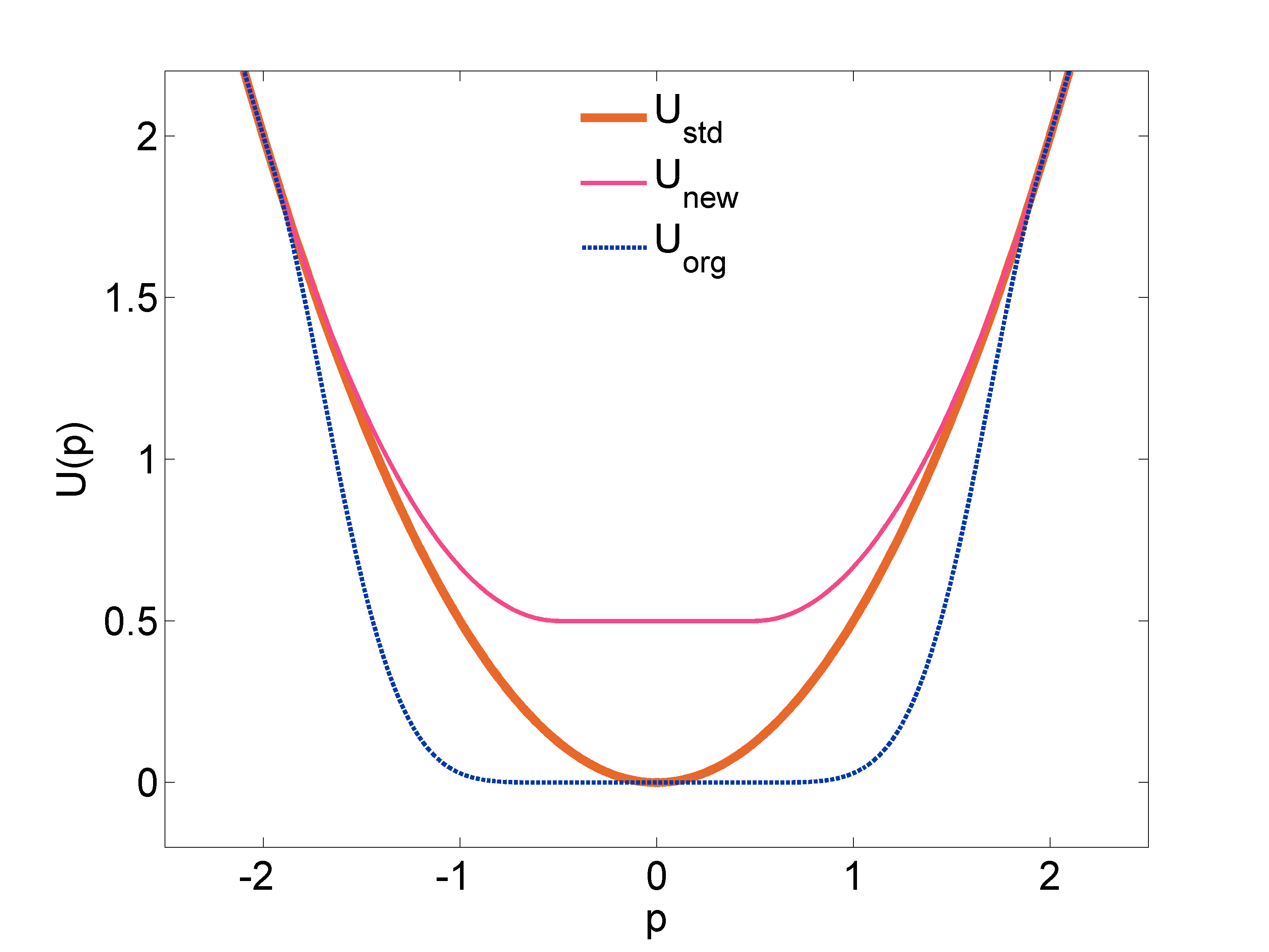}%{Figures/modifiedHamiltonian_parametrization.png}
%      \caption{The AR-kinetic energy function \eqref{def new arps kinetic energy} for various choice of parameters $\er$ and $\ef$.}
      \caption{Comparison of the AR-kinetic energy functions \eqref{def arps kinetic energy} and \eqref{def new arps kinetic energy}.}% for various choice of parameters $\er$ and $\ef$.}
      \label{fig:modifiedHamiltonian_parametrization}
    \end{subfigure}
    ~ %add desired spacing between images, e. g. ~, \quad, \qquad, \hfill etc. 
    %(or a blank line to force the subfigure onto a new line)
    \begin{subfigure}[b]{0.45\textwidth}
      \includegraphics[width=\textwidth]{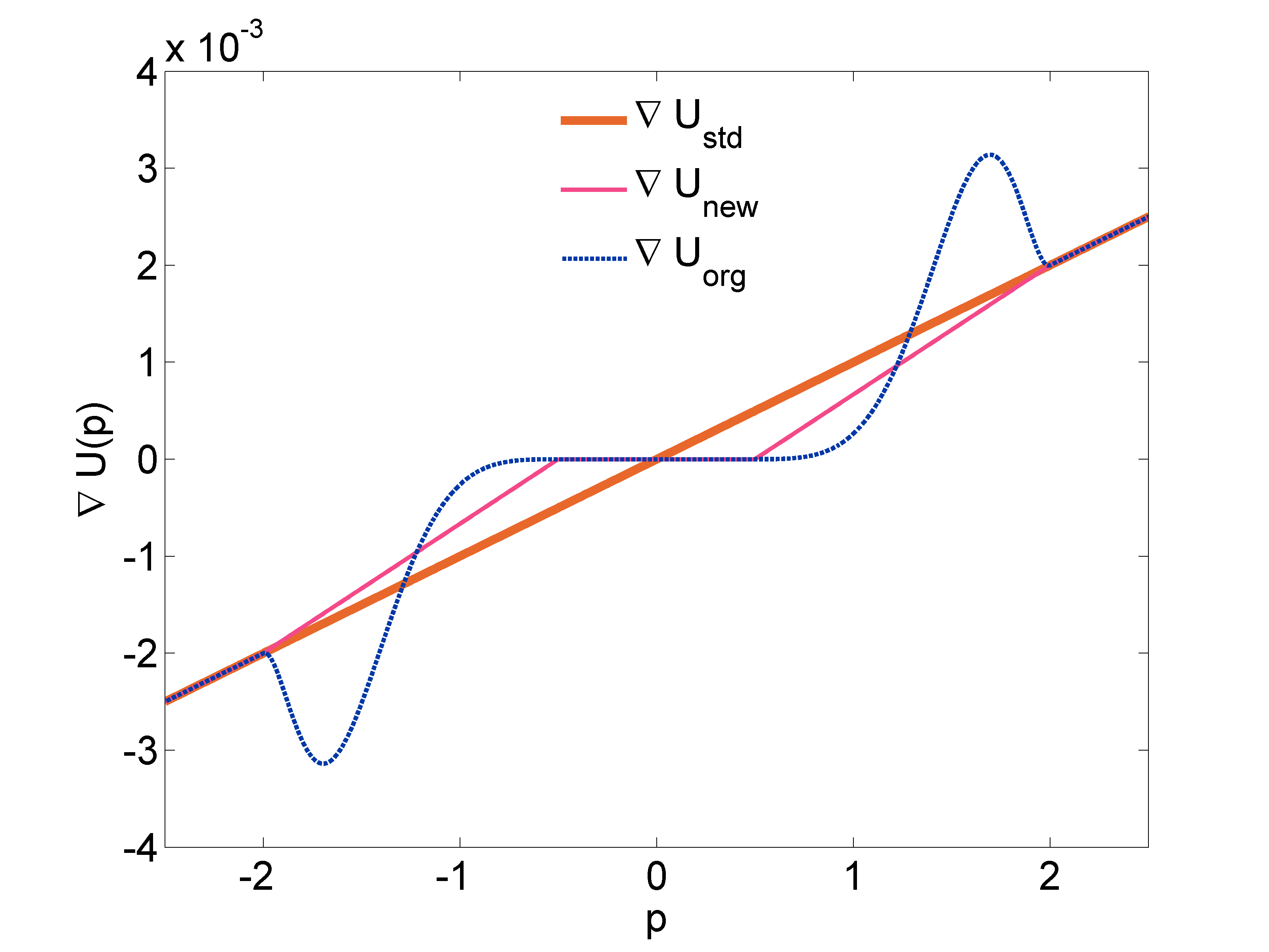}%{Figures/nablaU.png}
      \caption{Gradient interpolation of the kinetic energy ($U_{\rm new}$) versus function interpolation ($U_{\rm org}$).}
      \label{fig:nablaU}
    \end{subfigure}
    \caption{Comparison between the AR-kinetic energy function \eqref{def new arps kinetic energy} and the original AR kinetic energy~\eqref{def arps kinetic energy}.}
\end{figure}

%--------------------------------------------------------------------------
\subsubsection{Average rejection rates }
\label{sec:efficiency_U}

Since the AR-kinetic energy in general has derivatives larger than the ones of the standard kinetic energy, the timestep should be reduced in order to preserve the stability of the numerical method. We characterize in this section the possible reduction of the timestep due to the modification of the kinetic energy. As described in Section~\ref{Complete Generalized Hybrid Monte-Carlo scheme}, we metropolize the AR-Langevin dynamics by first integrating the Hamiltonian part with~\eqref{eq: hmc} and then the fluctuation-dissipation part with~\eqref{eq:HMC_like_proposal_for_Metropolis}. This corresponds to the evolution operator $P_\dt^{\rm GHMC} = P_\dt^{\rm Ham}P_\dt^{\rm FD}$.

Recall that the average rejection rate of the Hamiltonian and fluctuation/dissipation parts, namely (with expectations over $(q,p) \sim \mu$ and over the random variables used in the updates)
\[
\mathcal{R}^{\rm Ham}(\de t):=\E\left(1-A^{\rm Ham}_{\de t}(q,p)\right),
\qquad
\mathcal{R}^{\rm FD}(\de t):=\E\left[1-A^{\rm FD}_{\de t}(p,G)\right],
\] 
respectively scale as $\dt^3$ and $\de t^{3/2}$ (see Lemmas~\ref{lemma rejection rate hmc} and~\ref{lem:weak_FD}). We consider three kinds of AR-kinetic energies: the original function interpolation~\eqref{def arps kinetic energy}, and two interpolation functions~\eqref{def new arps kinetic energy} based on the gradient. More precisely, we either choose a linear spline or a $C^2$ spline by a polynomial of order~5 on the gradient $\nabla U$. The corresponding kinetic energies are respectively $C^2$ and $C^3$. The aim is to check the scaling of the rejection rates in terms of powers of $\dt$, and to estimate the prefactors for the various kinetic energies. 

We consider a system of $64$ particles of mass $m_i=1$ in a three dimensional periodic box with particle density $\rho=0.56$. The particles interact by a purely repulsive WCA pair potential, which is a truncated Lennard-Jones potential~\cite{SBB88}:
\[
V_{\rm WCA}(r) = \left \{ \begin{array}{cl}
  \displaystyle 4 \varepsilon_{\rm LJ} \left [ \left ( \frac{\sigma_{\rm LJ}}{r} \right )^{12}
    - \left ( \frac{\sigma_{\rm LJ}}{r}\right )^6 \right ] + \varepsilon_{\rm LJ} & \quad {\rm if \ } r \leq r_0, \\
  0 & \quad {\rm if \ } r > r_0,
\end{array} \right.
\]
where $r$ denotes the distance between two particles, $\varepsilon_{\rm LJ}$ and $\sigma_{\rm LJ}$ are two positive parameters and $r_0=2^{1/6}\sigma_{\rm LJ}$. In our simulations the parameters of the potential are set to $\e_{\rm LJ}=1, \sigma_{\rm LJ}=1$, while the parameters of the AR-Langevin dynamics \eqref{modified Langevin} are set to $\gamma=1,\beta=1$.

Figure~\ref{fig:Rejection_rate_FD} shows the average rejection rates for the AR parameters $\ef=2$ and $\er=1$ for $\UN$, as well as $\Ef=2$ and $\Er=0.5$ for $\UO$. This choice of parameters corresponds to $\sim 30\%$ percent of particles which are frozen for both AR-kinetic energies, \textit{i.e.} which are in the region where $\n U$ vanishes (see~\cite{trstanova2015speedUpARPS} for a thorough discussion on the link between the percentage of frozen particles and the algorithmic speed-up). Note that the predicted scalings of the rejection rates are recovered in all cases. The prefactor is however larger for the kinetic energy $\UO$ from~\cite{PRL-ARPS} than for $\UN$, especially for the fluctuation-dissipation part. The prefactor is also slightly smaller for the kinetic energy based on the gradient interpolation with a linear function, which is fortunate since $\n U$ has a lower computational cost than for interpolations based on higher order splines. 

\begin{figure}[t]
  \centering
  \begin{subfigure}[b]{0.45\textwidth}
    \includegraphics[width=\textwidth]{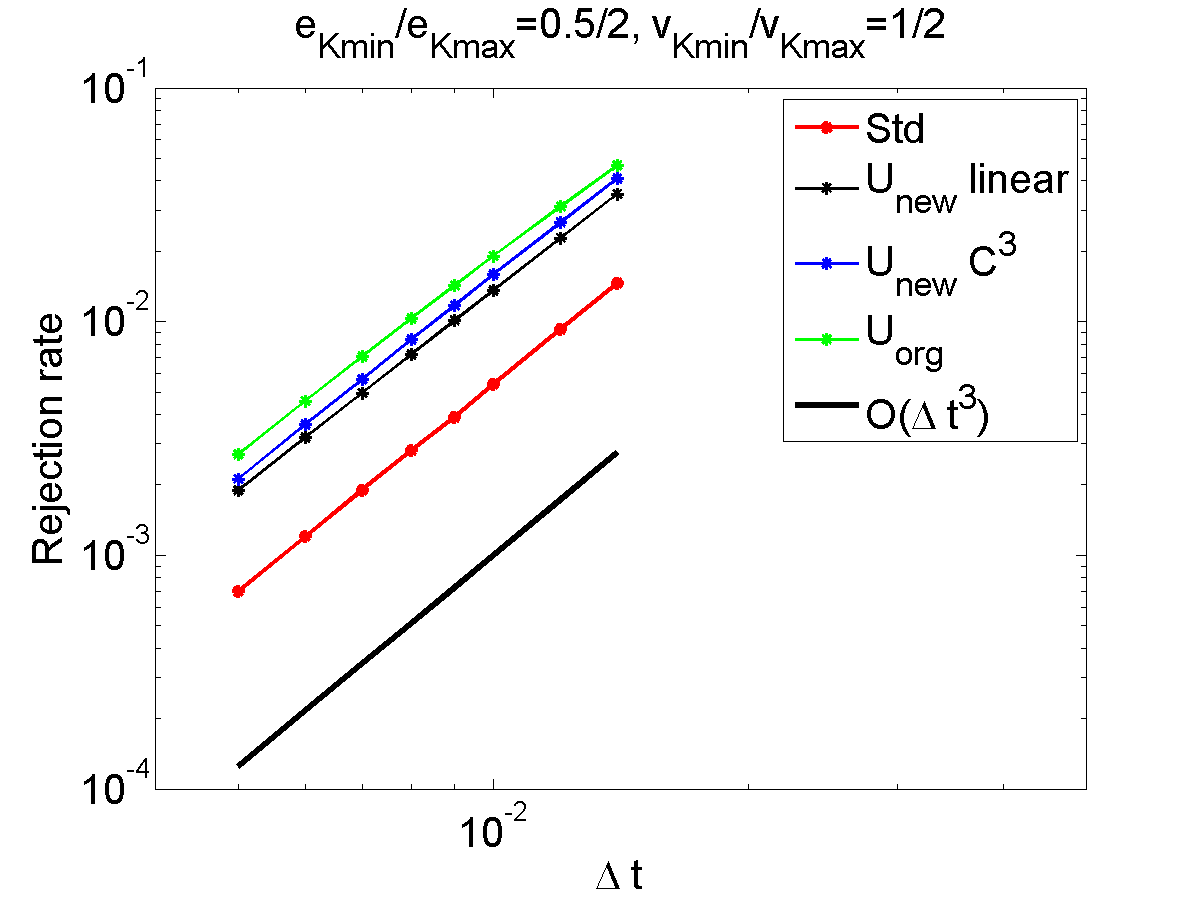}
    %{Figures/rejRatesHam_percRestr2_delta_0d4.png}
    %rejRatesHam_percRestr18_delta_0d4.png}%rejRatesHam6_delta_0d9.png}
    \caption{Hamiltonian part}
    \label{fig:rejRatesHam6_delta_0d9}
  \end{subfigure}
  ~ %add desired spacing between images, e. g. ~, \quad, \qquad, \hfill etc. 
  %(or a blank line to force the subfigure onto a new line)
  \begin{subfigure}[b]{0.45\textwidth}
    \includegraphics[width=\textwidth]{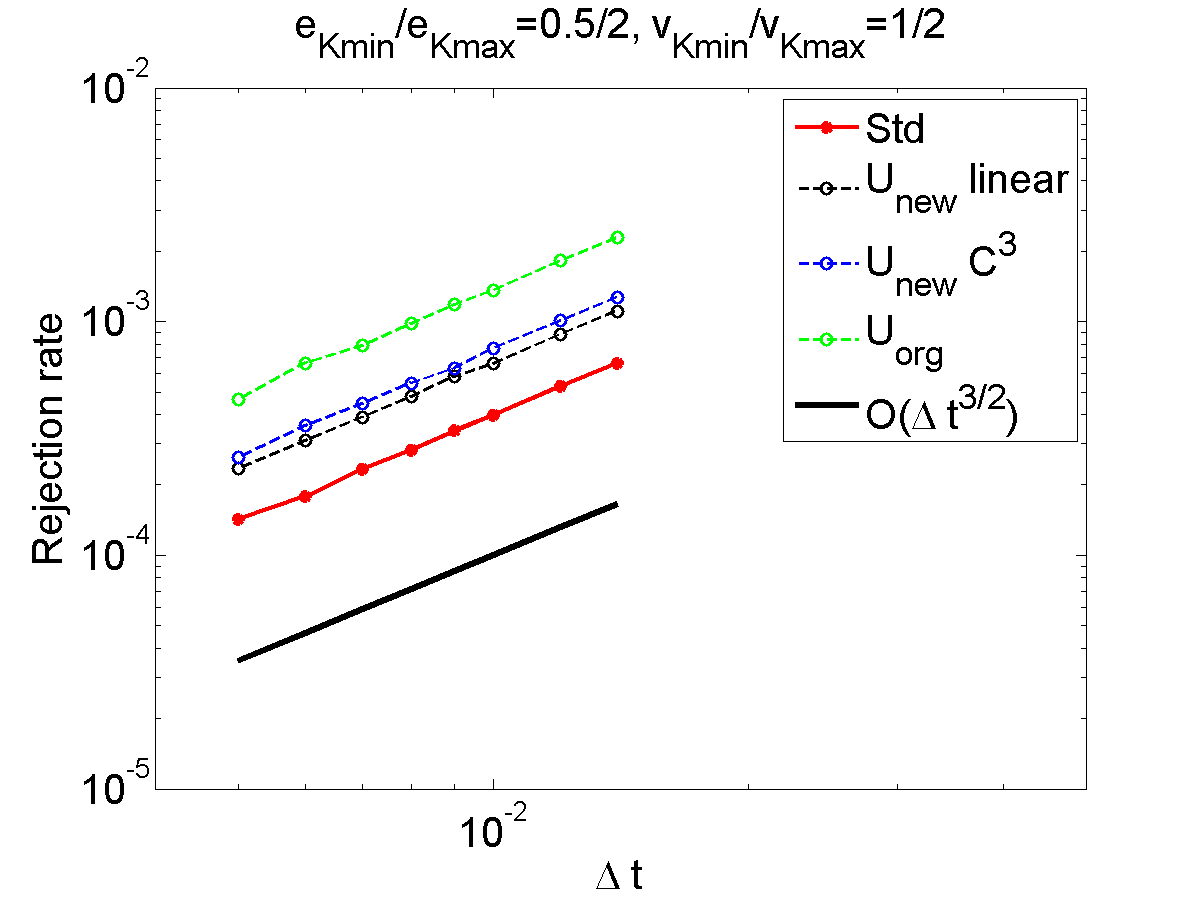}
    %{Figures/rejRatesFD_percRestr2_delta_0d4.png}%rejRatesFD_percRestr18_delta_0d4.png}%rejRatesFD6_delta_0d9.png}
    \caption{Fluctuation-dissipation part.}
    \label{fig:rejRatesFD6_delta_0d9}
  \end{subfigure}
  \caption{%REMARK: NOTE THAT HAM PRL IS NOT $DT^3$, consider also other $\er$?!
    Average rejection rates of GHMC as a function of the timestep for various kinetic energies (see text). The scaling of the rejection rates corresponds to the predicted orders, {\it i.e.} $\dt^3$ for the Hamiltonian part and $\dt^{3/2}$ for the fluctuation-dissipation part. }
  \label{fig:Rejection_rate_FD}
\end{figure}

It is possible to numerically determine the prefactor $C$ such that the rejection rate is approximately equal to $C\dt^\alpha$ (with $\alpha=3$ for the Hamiltonian part, and $\alpha=3/2$ for the fluctuation/dissipation). We refer to~\cite{PhD} for numerical evidence of improved properties of the new AR-kinetic energy function demonstrated by a reduced prefactor in the rejection rate of the GHMC scheme (see also Figure~\ref{fig:Rejection_rate_FD}).

We are now in position to determine the variations in the admissible timesteps as a function of the kinetic energies. We fix to this end a rejection rate, for the Hamiltonian part since this subdynamics mixes information on the positions and momenta, and involves the forces $-\n V(q)$ which are often at the origin of the stability limitations. Similar results are however obtained for the fluctuation/dissipation part, see~\cite{PhD}.

In our tests, we set the target rejection rate to two values: $\mathcal{R}^{\rm Ham}(\dt) \in \left\{ 0.001, 0.5\right\}$. Figure~\ref{fig:Dt} presents the timesteps $\dt$ achieving the desired rejection rates (normalized by $\de t_{\rm std}$, the timestep corresponding to the given rejection rate for the standard quadratic energy), for the kinetic energy $\UN$ (with an interpolation spline such that $\UN\in C^3$) and for various values of the parameters. We observe that the timestep should be reduced with respect to the standard case when the transition becomes somewhat sharper, \textit{i.e.} for $\delta:=\er/\ef$ approaching~1. Surprisingly, we observe that for smaller values of $\delta$, the timestep can in fact be increased compared to standard Langevin dynamics.

\begin{figure}[t]
    \centering
    \begin{subfigure}[b]{0.48\textwidth}
      \includegraphics[width=\textwidth]{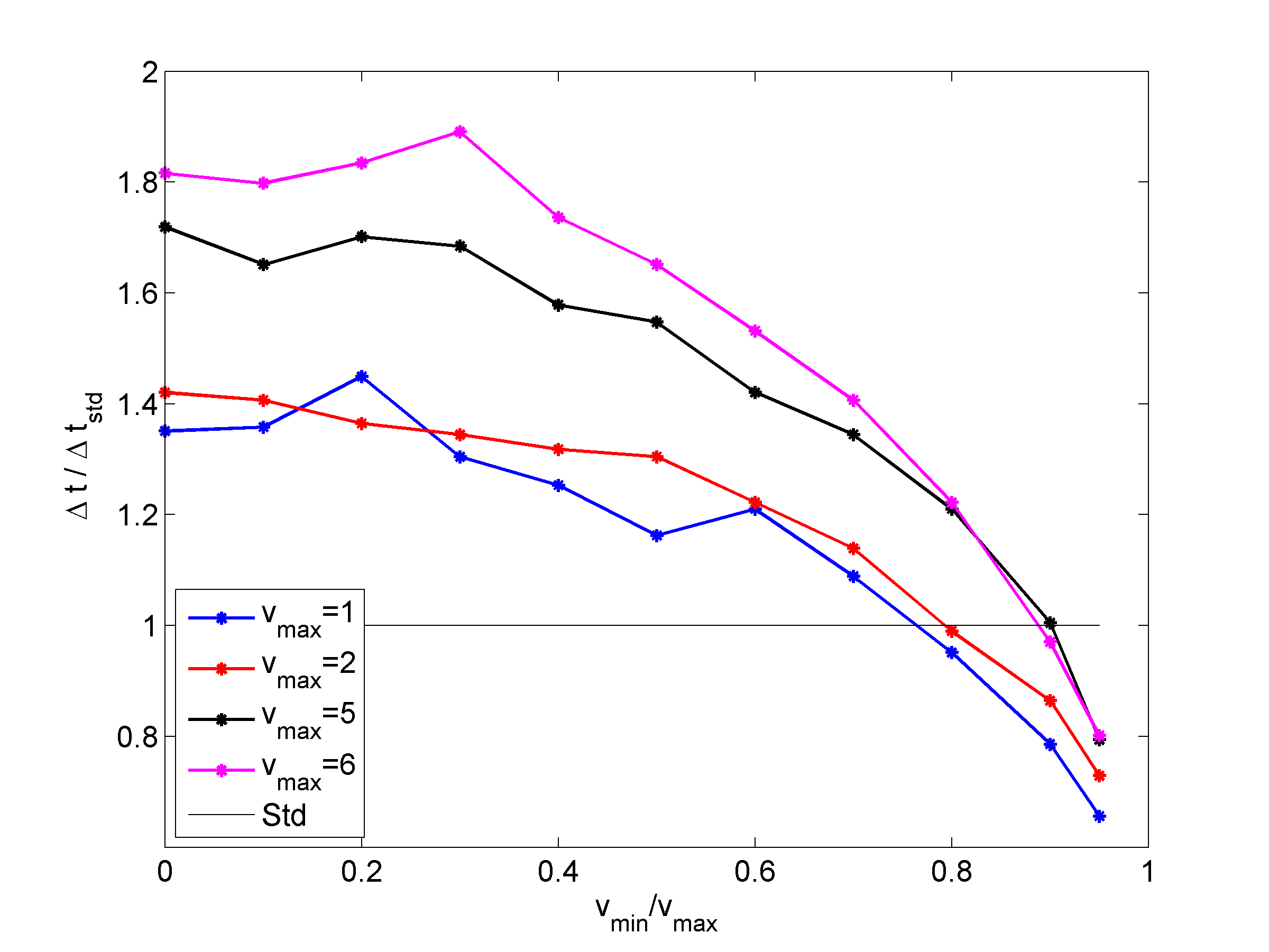}
      \caption{Rejection rate fixed at $0.001$}
      \label{fig:DtHam}
    \end{subfigure}
    %add desired spacing between images, e. g. ~, \quad, \qquad, \hfill etc. 
    %(or a blank line to force the subfigure onto a new line)
    \begin{subfigure}[b]{0.48\textwidth}
      \includegraphics[width=\textwidth]{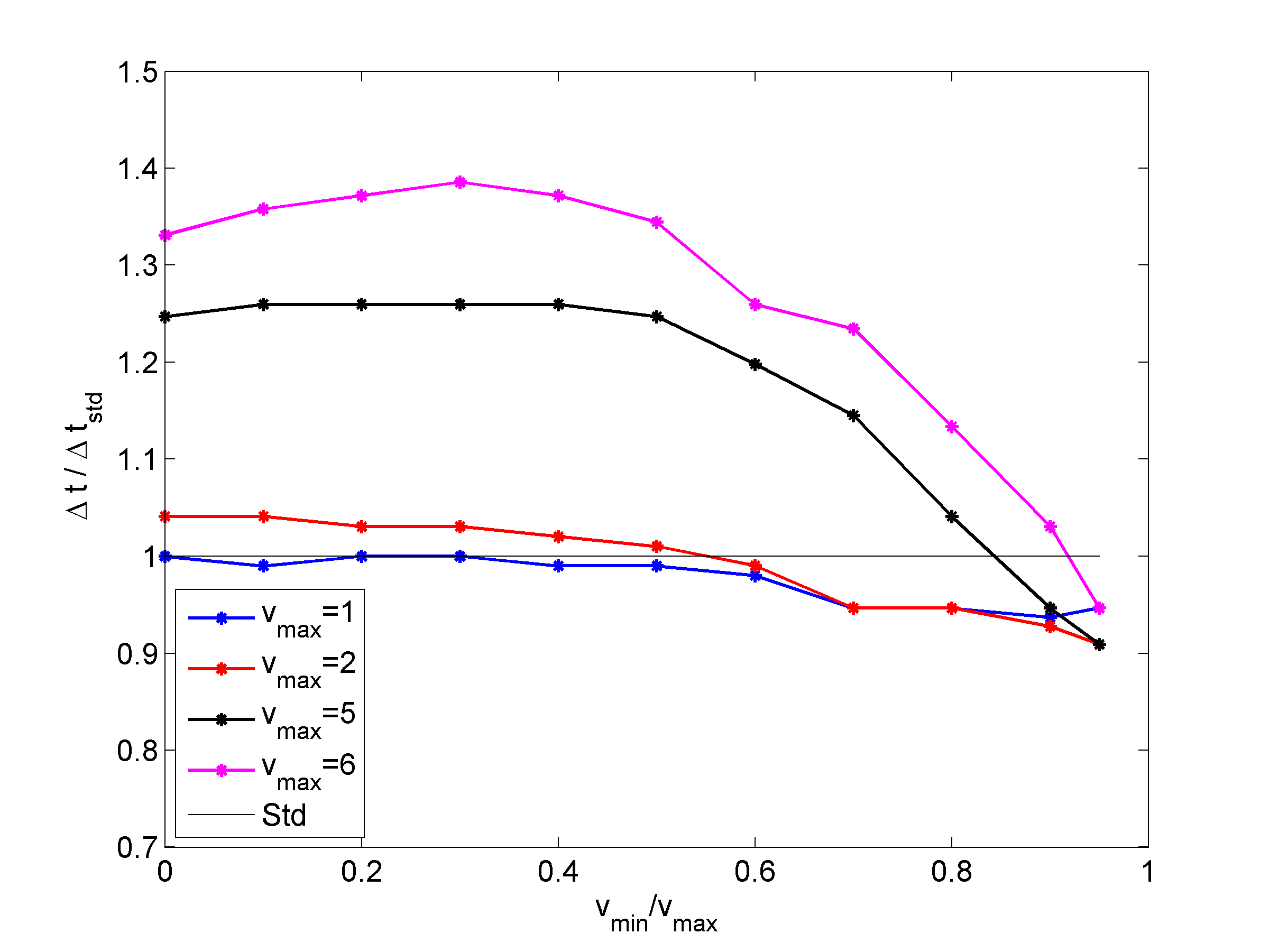}
      \caption{Rejection rate fixed at $0.5$}
      \label{fig:rejRate0d5}
    \end{subfigure}
    \caption{ Timesteps normalized by $\de t_{\rm std}$ (the time step corresponding to the same rejection rate for the standard kinetic energy) corresponding to a fixed rejection rate in the Hamiltonian part for various values of $\delta = \er/\ef$ and the kinetic energy~\eqref{def new arps kinetic energy}. }
      \label{fig:Dt}
\end{figure}

%% {  Moreover, we indeed observe reduced systematic error (bias) for the un-metropolized Langevin dynamics~\cite{PhD} for the AR-dynamics with $\ef=6, \er=0$ (see Figure~\ref{fig:bias}).  }
%\begin{figure}[t]
%    \centering
%        \includegraphics[width=0.5\textwidth]{Figures/bias.png}
%      \caption{ Long-time error on the expected value of the potential energy obtained by non-metropolized AR-Langevin dynamics with $\er=0$ and $\ef=6$ is smaller that when computed with the standard dynamics. }
%      \label{fig:bias}
%\end{figure}

%-------------------------------------------------------------------------------------------
\subsection{Decreasing metastability with general kinetic energies}
\label{sec:num_non_glob_Lip}

In this section, we illustrate how the use of alternative kinetic energy functions can help to reduce metastability in the sampling of probability measures of Boltzmann--Gibbs type. This possibility was already explored to some extent in recent works, using in particular relativistic kinetic energies with heavy tails~\cite{lu2016relativistic,livingstone2017kinetic,PhD}. We first present some numerical evidence showing that the finite time weak error is small with the scheme we consider, as exemplified by the computation of some total rate of escape out of some metastable state (see Section~\ref{sec:weakErrorNumerics}). In a second step, we study the reduction of metastability incurred by appropriate choices of kinetic energies (see Section~\ref{section hitting times}).

\subsubsection{Improved weak order for approximation of dynamical properties}
\label{sec:weakErrorNumerics}

We illustrate the weak error estimate~\eqref{eq:finite_time_weak_error_GHMC} obtained as a corollary of Proposition~\ref{lemma evolution operator ghmc} in two steps. First, we consider a situation where we can analytically integrate the fluctuation-dissipation in order to confirm the fractional order~$3/2$ following from the estimates of Lemma~\ref{lem:weak_FD}. In a second step, we show that the use of the complete GHMC scheme allows to reduce the weak error compared to schemes using the standard MALA discretization of the fluctuation-dissipation~\cite{RDF78,RT96} (based on a proposal obtained by a Euler--Maruyama discretization).

\paragraph{Confirmation of the fractional order~$3/2$ for the fluctuation-dissipation.}
We consider the elementary fluctuation-dissipation dynamics~\eqref{eq:elementary_FD} for the standard kinetic energy $U(p)=\frac{p^2}{2m}$, in dimension $d=1$. We compute for instance the variance of $p_T$ starting from the initial momentum $p^0=0$, for a given time $T>0$. An analytical integration of the Ornstein--Uhlenbeck process gives
\[
\E \left[ p_T^2\right | p^0=0] = \beta^{-1} \left(1-{\rm e}^{-\frac{2 T}{m}}\right).
\]
We compare the discretization of~\eqref{eq:elementary_FD} by the HMC-like scheme~\eqref{eq:HMC_like_proposal_for_Metropolis} and by MALA. We set $T=1$, $m=1 $, $\beta=1$, and approximate the expectation by a sum over $10^8$ realizations. The timesteps are chosen in the range $[0.005, 0.02]$. The relative errors reported in Figure~\ref{fig:varOU} confirm the predicted order~$3/2$ for HMC, as well as the order~1 for MALA (obtained as a straightforward corollary of the results in~\cite{FHS14,fathi2015improving}). 

\begin{figure}[t]
  \centering
  \includegraphics[width=0.45\textwidth]{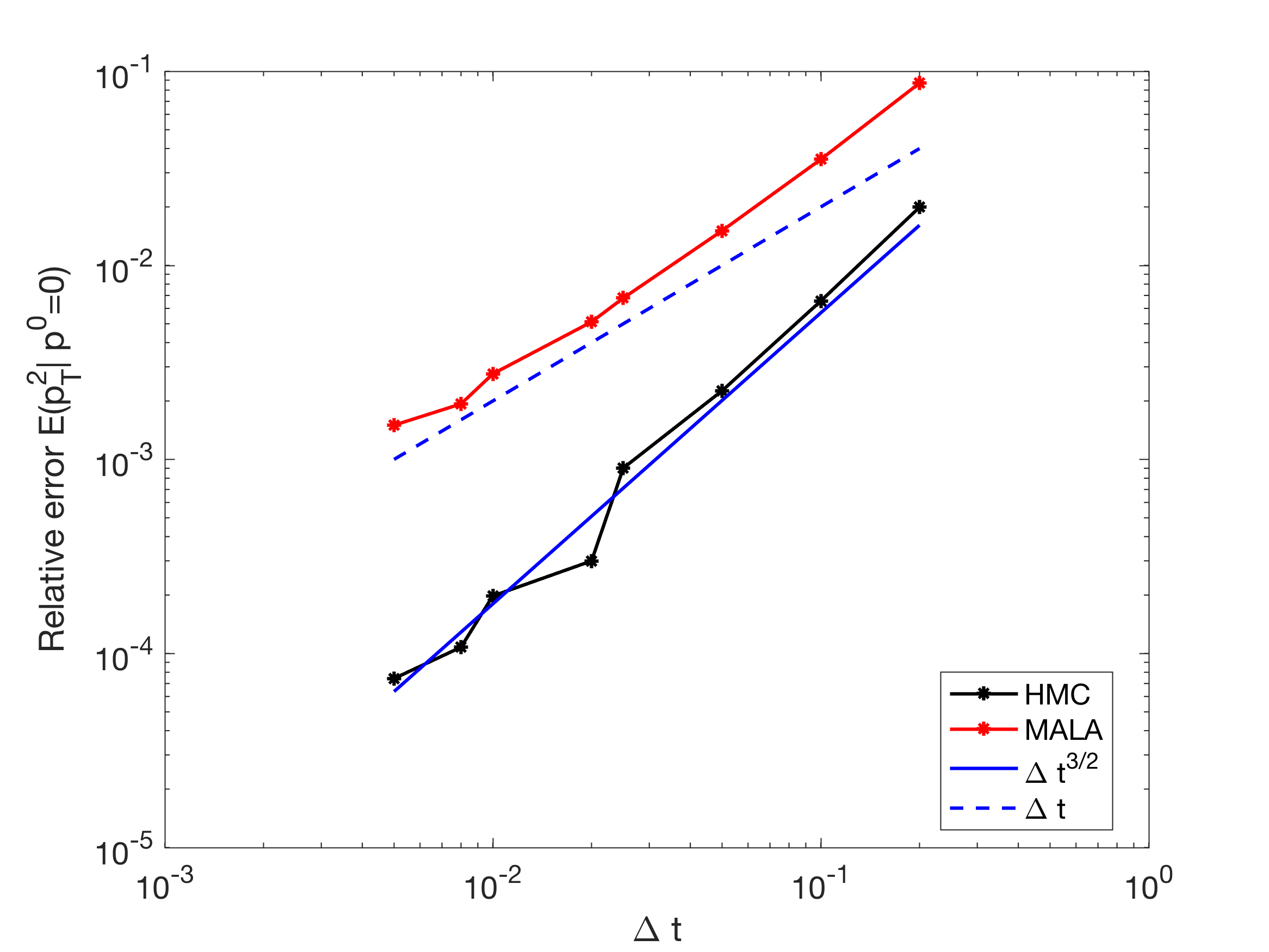}
  \caption{Relative error of the variance of momenta at a given time $T=1$ for the Ornstein--Uhlenbeck process.}
  \label{fig:varOU}
\end{figure}

\paragraph{Weak order for the full dynamics.}
We next consider the full Langevin dynamics in dimension $d=1$, with a potential energy given by a double-well potential: for $q\in \R$,
\[
V(q) = (q^2-1)^2.
\]
We consider two kinetic energies: the standard, quadratic one with $m=1$, and a generalized kinetic energy $U=V$. We apply the new scheme corresponding to the Strang splitting encoded by the evolution operator $P^{\rm GHMC}_{\de t}= P^{\rm FD}_{\de t/2}  P^{\rm Ham}_{\de t}  P^{\rm FD}_{\de t/2}$. In order to illustrate the improved weak error predicted by Proposition~\ref{lemma evolution operator ghmc}, we compute the probability that a trajectory starting from the initial condition $q_0=1$ is within the other metastable set, \emph{i.e.} behind the saddle point of the double-well, at a fixed time $T=2$. More precisely, we approximate $\E\left[\1_{A(q_T)} \left| q^0=q_0 \right.\right]$ with $A(q)=\{ q \in \mathbb{R}, \, q<0\}$ using $10^8$ independent realizations, for various time steps $\de t$. We compare two methods: the here proposed GHMC and GMALA, which is obtained by the same Strang splitting, but with a fluctuation-dissipation discretized by MALA. Note that the weak order of GMALA is~1 since the fluctuation-dissipation is discretized at order~1 only. From the numerical results reported in Figure~\ref{fig:1A} we observe that GHMC method is more accurate than GMALA, especially for kinetic energies different from the standard quadratic one. As expected, the finite time weak error is of order~$\dt$ for GMALA; for the GHMC scheme we consider the order is not completely clear, but the error is in any case of order~$\dt^{3/2}$ at most.  

Note also that, interestingly, the probability of being in the set $A$ at time $T=2$ is larger with the alternative kinetic energy $U=V$ than with the standard one ($0.22$ versus $0.12$), which suggests that Langevin dynamics with this choice of kinetic energy explores the phase space faster. We continue exploring this idea in Section~\ref{section hitting times} below.

\begin{figure}[t]
  \centering
  \includegraphics[width=0.45\textwidth]{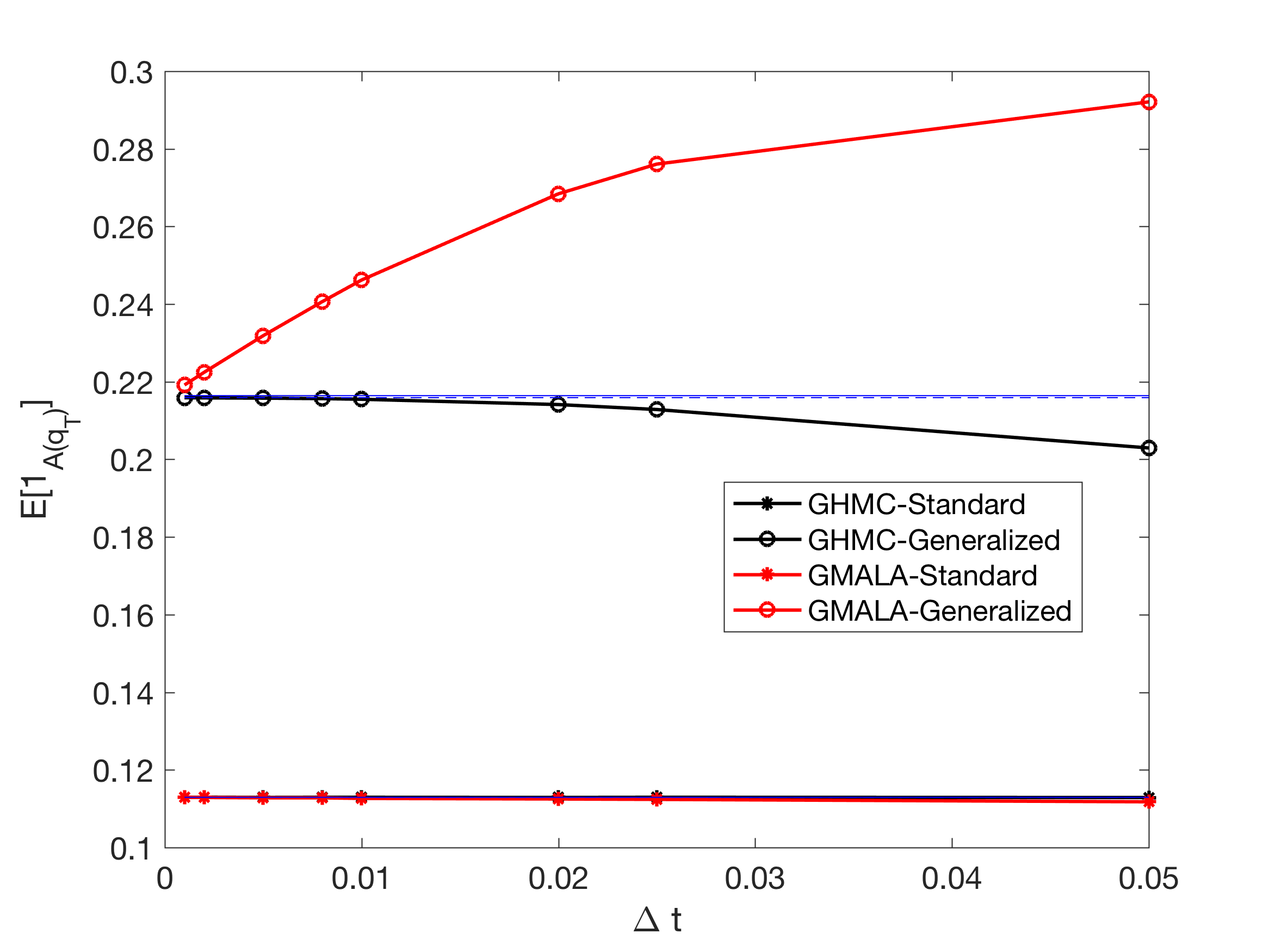}
  \caption{Probability of hitting the set $A$ at time $T=2$ for generalized and standard kinetic energy, as a function of the timestep $\Delta t$.}
  \label{fig:1A}
\end{figure}

\begin{figure}[t]
  \centering
  \includegraphics[width=0.45\textwidth]{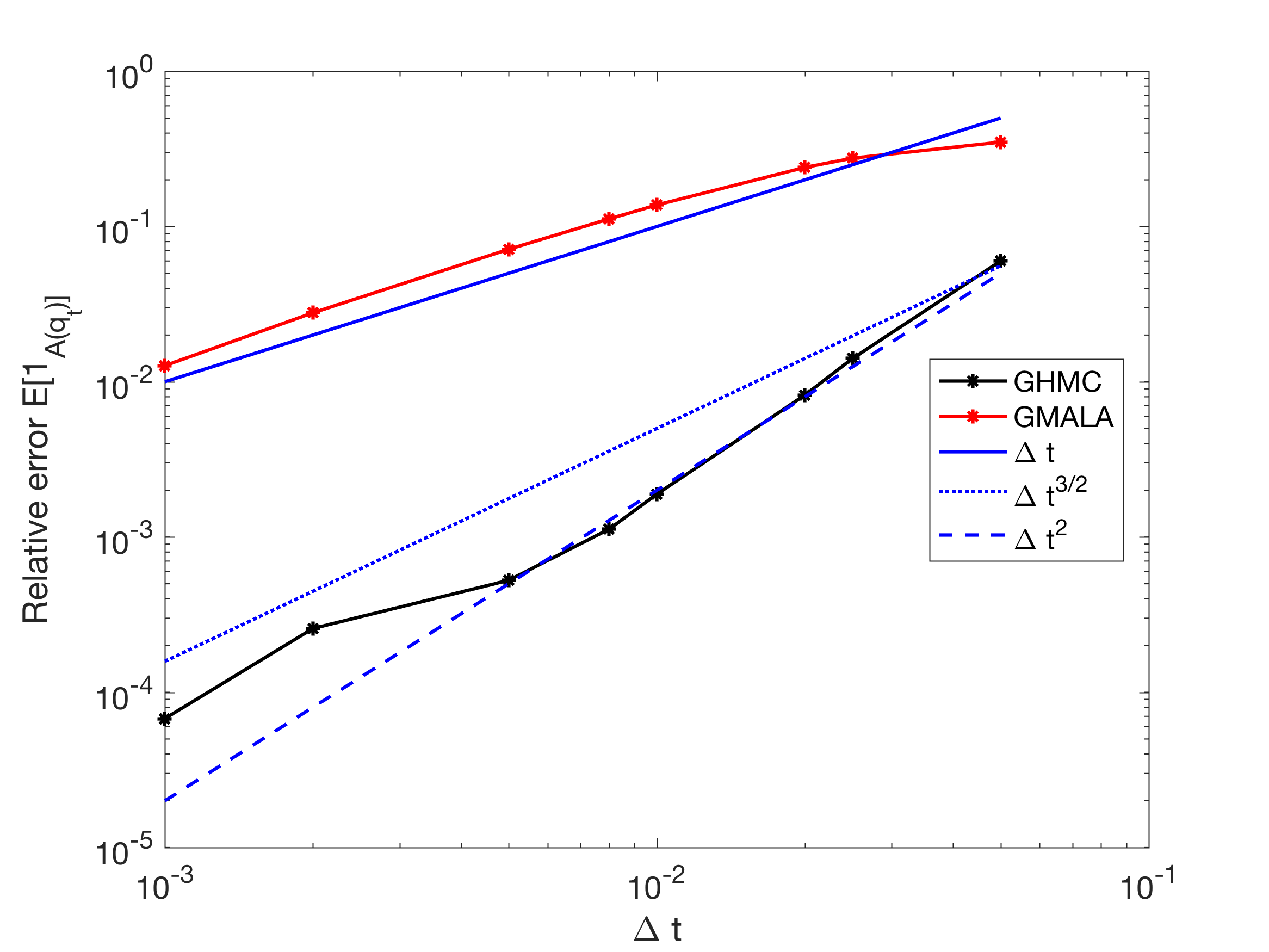}
  \caption{Relative error of the probability of hitting the set $A$ at time $T=2$ with respect to the interpolated values at $\de t=0$ for the generalized kinetic energy. As expected, the finite time weak error if of order $\de t$ for GMALA; for the GHMC scheme we consider, the order is not completely clear, but the error is in any case of order $\de t^{3/2}$ at most. }
  \label{fig:modif}
\end{figure}

%We have demonstrated Proposition~\ref{lemma evolution operator ghmc}.
%The results from the second example suggest that the alternative kinetic energy has dynamical properties, which allow better exploration of the phase space. This property is very important for sampling of the invariant measure. Because the method is metropolized and the positional marginal distribution does not change, this scheme can be used to compute averages with respect to the positions in shorter simulation time. 

\subsubsection{Expected hitting times}
\label{section hitting times}

Motivated by the numerical results reported in the previous section, which demonstrate an improved exploration of the phase space for the special case $U=V$, we look at other kinetic energies and their impact on the metastable features of the dynamics. We consider, as a measure of the metastability, the average hitting time of a metastable state starting from another metastable state (see below for more precise definitions). The numerical approximation of such quantities is not dictated by weak error estimates, but rather by strong error estimates. Our aim in this section is however rather to explore properties of the underlying continuous dynamics, so that we consider sufficiently small timesteps and neglect the impact of the discretization error. 

We study two dimensional systems (\textit{i.e} $q=\left(x,y\right)\in \R^2$) for a potential similar to the one considered in~\cite[Section~1.3.3.1]{lelievre2010free}: 
\begin{equation}
\label{eq:2D_DW}
V(x,y)=\frac{1}{6}\left(4\left(-x^2-y^2+w\right)^2+10\left(x^2-2\right)^2+\left(\left(x+y\right)^2-1\right)^2+\left(\left(x-y\right)^2-1\right)^2\right).
\end{equation}
This potential can be seen as some effective double well potential in the $x$ direction (see Figure~\ref{fig:potetnial_and_hitting_set} for contour plots). 
\begin{figure}[t]
  \centering
  \includegraphics[width=0.45\textwidth]{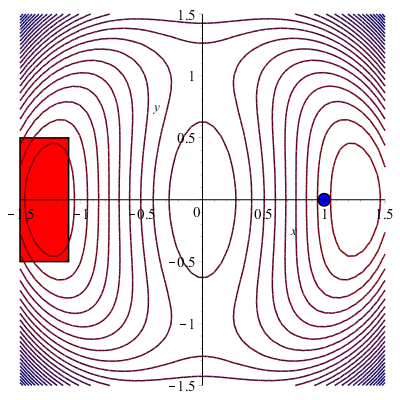}
  \caption{Two dimensional double-well potential~\eqref{eq:2D_DW}. To compute exit times out of metastable states, we consider the starting configuration $A:=(1,0)$ and the target set $B:=\left\{(x,y): x\leq -1 \text{ and } \left|y\right|\leq 0.5\right\}$.}
  \label{fig:potetnial_and_hitting_set}
\end{figure}
The metastability of Langevin dynamics is caused by some energetic barrier in this direction at $x=0$. In the following numerical experiments, we discretize the Langevin dynamics \eqref{modified Langevin} by the same scheme as in Section~\ref{section AR-Langevin dynamics}, with $\gamma=1$, $m=1$ and $\de t=0.001$. 

Various kinetic energies can be considered. We focus on the following ones: 
\begin{enumerate}[(1)]
\item the standard kinetic energy $U_1(x,y)=(x^2+y^2)/2$;
\item a fifth order polynomial in both directions $U_2(x,y)= \left( \left|x\right|^5+\left|y\right|^5\right)/5$, which provides an example of light-tailed distribution of momenta;
\item a heavy tailed function distribution of momenta, corresponding to the choice 
\[
U_3(x,y)= \frac{4}{5} \left[ \left|x\right|^{5/4} +\left|y\right|^{5/4} \right]; 
\]
\item the same function as the potential function $U_{4}\equiv V$;
\item a double-well function in the $x$-direction and a quadratic function in the $y-$direction:  
\[
U_5(x,y) = V_{\rm DW}(x) + \frac{y^2}{2},
\qquad
V_{\rm DW}(x)=\left( \left| x-1 \right|^{-2} + \left| x+1 \right| ^{-2} \right) ^{-1}.
\]
This function somewhat approximates $V$, so we expect the distribution of momenta under the canonical measure associated with $U_5$ to be close to the one associated with $U_4$. 
\end{enumerate}
Figure~\ref{fig:animals} presents two realizations of the Langevin dynamics~\eqref{modified Langevin} for a physical time $T=1000$ and an inverse temperature $\beta = 1$, for the choices $U_1$ and $U_4$ above. Note that, for the standard kinetic energy~$U_1$, there is only one crossing from one well to the other during the simulation time. On the other hand, there are many more crossings for $U_{4}$. 
\begin{figure}
    \centering
    \begin{subfigure}[b]{0.45\textwidth}
      \includegraphics[width=\textwidth]{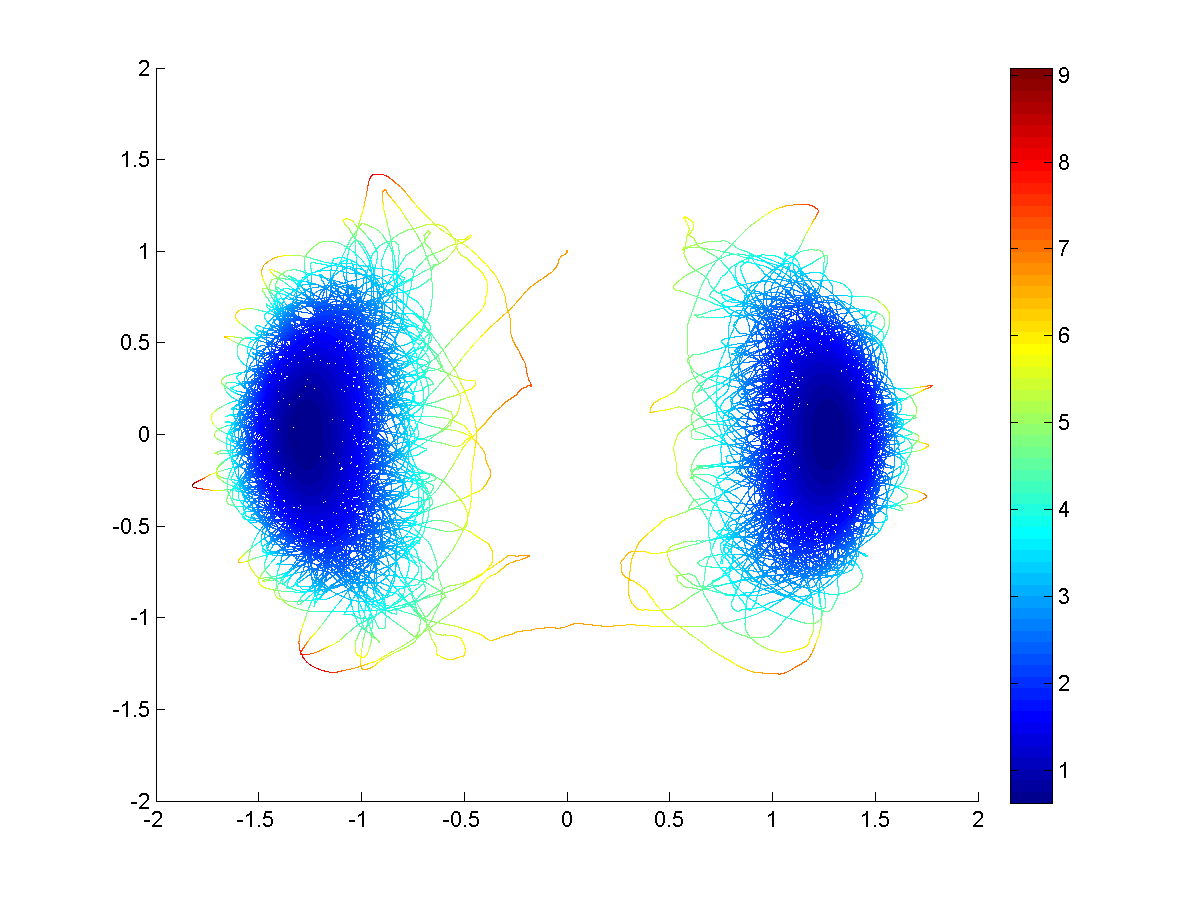}
      \caption{Standard kinetic energy function. \vspace{12pt}}
    \end{subfigure}
    \quad 
    \begin{subfigure}[b]{0.45\textwidth}
      \includegraphics[width=\textwidth]{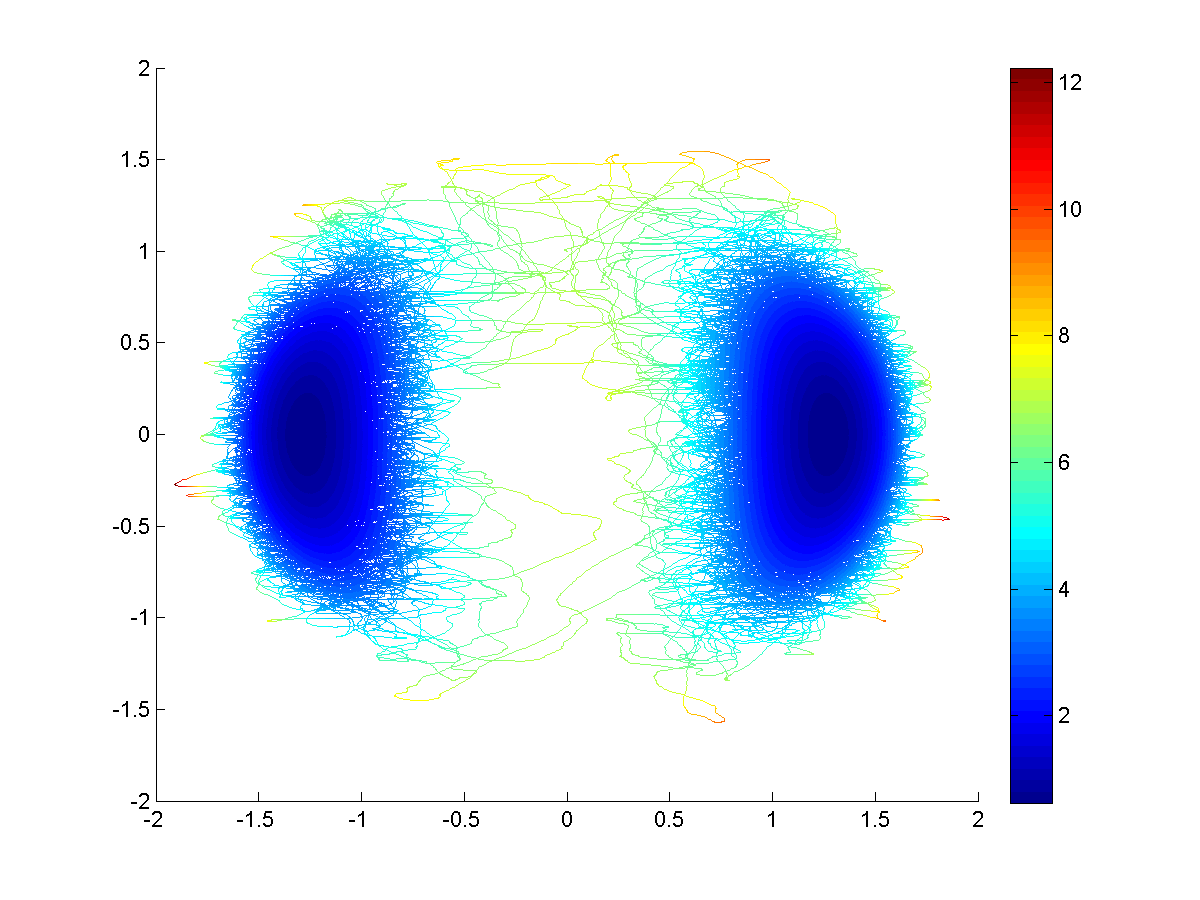}
      \caption{Same kinetic energy function as the potential energy function, \textit{i.e.} $U\equiv V$.}
    \end{subfigure}
    \caption{Positions as a function of time for the modified Langevin dynamics with the two-dimensional double well potential~\eqref{eq:2D_DW}, and two different kinetic energy functions. The simulation time is $T=1000$, %$t =0.001$,
		and the same realization of the Brownian motion is used in both cases. For the same number of simulation steps, there are more crossings between the wells for the dynamics with the modified kinetic energy (Right) than for the standard one (Left). The coloring corresponds to the values of the potential energy.}
    \label{fig:animals}
\end{figure}

In order to quantify the reduction of the metastability gained by modifying the kinetic energy function, we numerically estimate the expected time to reach a set $B$ starting from a set~$A$, the two sets being separated by the energetic barrier. We start in fact from a given initial condition, which corresponds to the initial set $A:=\{(1,0)\}$. We then compute the number of simulation steps necessary to reach the set $B:=\left\{(x,y): x\leq -1 \text{ and } \left|y\right|\leq 0.5\right\}$ (see Figure~\ref{fig:potetnial_and_hitting_set} for an illustration). The expected hitting time is estimated by an average over $1000$ independent realizations of the exit process. We report in Table~\ref{table hitting times} the average physical time needed to reach the set $B$ for each choice of the kinetic energy function, as well as the speed-up relative to the results obtained with the standard kinetic energy. 
\begin{table*}[h]
  \centering
  \begin{tabular}{|c||c|c |c|c|c|}
    \hline
    Kinetic energy &$U_1=U_{\rm std}$ &   $U_2$  &  $U_3$  &  $U_4$  &  $U_5$ \\% &  $U_6$ \\
    \hline
    $T_{\rm hit}$&$297.2 \left[\pm 9.5\right]$ &    $259.2  \left[\pm 7.8\right]  $ &$ 307.0\left[\pm 9.6 \right]$ & $ 101.7 \left[\pm 3.2 \right]  $ & $ 203.4\left[\pm 6.3 \right] $  \\
    \hline
    Speed up $T_{\rm hit}/T_{\rm std}$  &  $1$ &   $  1.155   $ & $ 0.97  $ &$ 2.92  $& $1.46$     \\
    \hline
  \end{tabular}
  \caption{Expected hitting times according to the choice of the kinetic energy functions $U_i$ (see text) at $\beta=1$. Errors bars determined by 95\% confidence intervals are reported in brackets.} 
  \label{table hitting times}
\end{table*}
Intuitively, heavy tailed distributions of momenta (corresponding to~$U_3$ here) could be thought of as being interesting since they allow for larger velocities, which may facilitate the transition from one well to the other. This is however not the case. On the other hand, we observe that the double-well-like functions ($U_4$ and $U_5$) are most helpful to reduce the metastability of the dynamics and allow for more transitions from the region around $x=-1$ to the region around $x=1$. Note that the hitting time is almost three times smaller with $U_4$. 

We next study the scaling of the average time needed to reach the set $B$ as a function of the inverse temperature~$\beta$, for the standard kinetic energy and the one which performed best at $\beta=1$, namely~$U_4 = V$; see Figure~\ref{fig:exitTimesbeta}. We observe an exponential growth of the hitting time with respect to~$\beta$, which is characteristic for metastability caused by energetic barriers in the low temperature limit by the Eyring-Kramers law (see for instance the presentation and the references in~\cite{berglund-13,actaLelievre2016}). We fit the hitting times as
\[
T_{\rm hit}(\beta) = C \mathrm{e}^{\beta E},
\]
for some energy~$E > 0$. For the results presented in Figure~\ref{fig:exitTimesbeta}, $E$ is the same for both kinetic energies, but the prefactor~$C$ differs. It is in fact smaller for the modified kinetic energy~$U_4$ than for the standard kinetic energy~$U_1$.

\begin{figure}[t]
  \centering
  \includegraphics[width=0.60\textwidth]{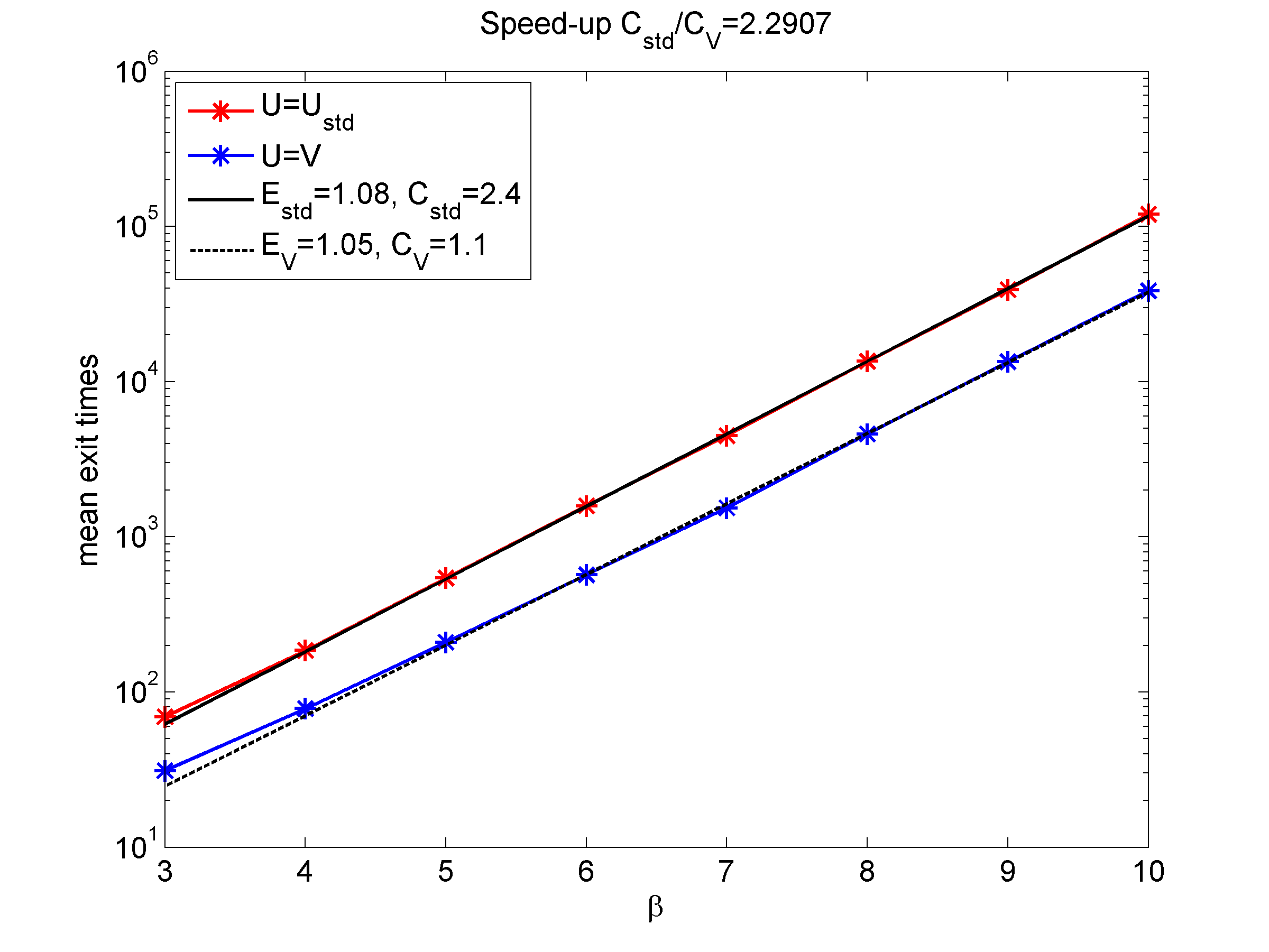}%exitTimesbeta.png}
  \caption{Mean exit times over $2000$ realizations as a function of $\beta\in \{3,4,5,6,7,8,9,10\}$.}%$100$ realizations as a function of\{0.1 ,0.3,0.5,0.7, 0.9, 1.5\}$.}
  \label{fig:exitTimesbeta}
\end{figure}

The excellent reduction in metastability we obtain on this simple low-dimensional system motivates us to test the relevance of this approch for higher dimensional systems. One track is to modify the kinetic energy on the velocity of some reaction coordinate summarizing slow degrees of freedom, keeping the standard kinetic energy for faster degrees of freedom; see~\cite{PhD} for preliminary steps in this direction.

%------------------ APPENDIX -------------------
\appendix
\section{Technical results used in the proof of Theorem~\ref{theorem DMS}}
\label{sec:technical_results}

Let us first gather some properties of the operator~$A$, directly deduced from~\cite[Lemma~1]{DMS15}. The proof is obtained by a direct adaption of the proof of~\cite[Lemma~1]{RS17}.

\begin{lemma}
\label{lem:A}
It holds $\Pi A = A$. Moreover, for any function $g \in L^2(\mu)$,
\[
\| A g \|_{L^2(\mu)} \leq \frac12 \| (1-\Pi)g \|_{L^2(\mu)}, 
\qquad 
\| \cLham A g \|_{L^2(\mu)} \leq \| (1-\Pi)g \|_{L^2(\mu)}, 
\]
\end{lemma}

Let us now turn to the proof of Proposition~\ref{prop:coercivity_scrD}, written for real-valued functions. Its proof is very similar to the proof of~\cite[Proposition~1]{RS17}, with a few modifications except for the estimate given in Lemma~\ref{lem:ALFD} below which requires a more involved treatment. First, note that 
\[
\begin{aligned}
\langle\langle \cL^* g,g \rangle\rangle & = \langle\langle g, \cL^* g \rangle\rangle = \langle g , \cL g \rangle_{L^2(\mu)} + \eps \langle A\cL^* g , g \rangle_{L^2(\mu)} + \eps \langle \cL A g , g\rangle_{L^2(\mu)} \\
& = \gamma \langle g , \cLFD g \rangle_{L^2(\mu)} - \eps \langle A \cLham g , g \rangle_{L^2(\mu)} + \gamma \eps \langle A \cLFD g , g \rangle_{L^2(\mu)} + \eps \langle \cLham A g , g \rangle_{L^2(\mu)} ,
\end{aligned}
\]
where we used in the last line that $\cLFD A = \cLFD \Pi A = 0$. Since $\cLFD = -\beta^{-1} \nabla_p^* \nabla_p$ (with $\nabla_p^* = -\nabla_p + \beta \nabla U^T$) and using Lemma~\ref{lem:A},
\begin{equation}
\label{eq:derivee_sH}
\begin{aligned}
  \langle\langle \cL^* g,g \rangle\rangle & \leq -\frac{\gamma}{\beta} \left\|\nabla_p g \right\|_{L^2(\mu)}^2 - \eps \langle A \cLham g , g \rangle_{L^2(\mu)} + \gamma \eps \langle A \cLFD g , g \rangle_{L^2(\mu)}  \\
  & \ \ + \eps \| (1-\Pi)g\|_{L^2(\mu)}^2.
\end{aligned}
\end{equation}
The first term on the right-hand side can be bounded using the Poincar\'e inequality on~$\kappa$: 
\[
-\frac1\beta \| \nabla_p g \|_{L^2(\mu)}^2 \leq -\frac{K_\kappa^2}{\beta} \| (1-\Pi)g \|_{L^2(\mu)}^2.
\]
The third term on the first line of the right-hand side is bounded using Lemma~\ref{lem:ALFD}. We next decompose the second term in the first line of the right-hand side of~\eqref{eq:derivee_sH} as 
\begin{equation}
\label{eq:ALham}
\langle A \cLham g , g \rangle_{L^2(\mu)} = \langle A \cLham \Pi g , g \rangle_{L^2(\mu)} + \langle A \cLham (1-\Pi) g , g \rangle_{L^2(\mu)}.
\end{equation}
We start with the first term on the right-hand side of the above equality. Denoting by $B = \cLham \Pi$, it holds $(B h)(q,p) = \nabla U(p)^T \nabla_q (\Pi h)(q)$. When $h \in L^2_0(\mu)$, the Poincar\'e inequality~\eqref{eq:Poincare} therefore leads to
%\[
%\| Bh \|_{L^2(\mu)} = \| \nabla U \|_{L^2(\kappa)} \| \nabla_q (\Pi h) \|_{L^2(\nu)} \geq K_\nu \| \nabla U \|_{L^2(\kappa)} \| \Pi h \|_{L^2(\nu)} = K_\nu \| \nabla U \|_{L^2(\kappa)} \| \Pi h \|_{L^2(\mu)}.
%\]
\[
\begin{aligned}
\| Bh \|^2 & = \int_{\mathcal{E}} \left(\sum_{i=1}^d \partial_{p_i} U \partial_{q_i}(\Pi h) \right)^2 d\mu \\
& = \frac1\beta \int_{\mathcal{D}^d} \nabla_q (\Pi h)^T \mathcal{M} \nabla_q(\Pi h) \, d\nu \\
& \geq \alpha \left\| \nabla_q (\Pi h) \right\|_{L^2(\nu)}^2 \geq \alpha K_\nu^2 \| \Pi h \|_{L^2(\nu)}^2 = \alpha K_\nu^2 \| \Pi h \|_{L^2(\mu)}^2,
\end{aligned}
\]
for some $\alpha > 0$, since the matrix 
\[
\cM = \int_{\R^d} \nabla^2 U \, d\kappa = \beta \int_{\R^d} \nabla U \otimes \nabla U \, d\kappa
\]
is positive, in view of the second equality, and in fact definite positive since $\nabla U \neq 0$ (otherwise $\rme^{-\beta U}$ would not be integrable). 
This can be rephrased as
\[
%B^* B \geq K^2_\nu \| \nabla U \|_{L^2(\kappa)}^2 \Pi \geq 0
B^* B \geq \alpha K^2_\nu  \Pi \geq 0
\]
in the sense of symmetric operators. Since $A\cLham \Pi = (1+B^*B)^{-1} B^* B$, we can conclude that
\[
%- \langle A \cLham \Pi g , g \rangle_{L^2(\mu)} \leq -\frac{K^2_\nu \| \nabla U \|^2_{L^2(\kappa)}}{1 + K^2_\nu \| \nabla U \|^2_{L^2(\kappa)}} \| \Pi g \|_{L^2(\mu)}^2.
- \langle A \cLham \Pi g , g \rangle_{L^2(\mu)} \leq -\frac{\alpha K_\nu^2 }{1 + \alpha K_\nu^2 } \| \Pi g \|_{L^2(\mu)}^2.
\]
For the second term on the right-hand side of~\eqref{eq:ALham}, we write (using $\Pi A = A$)
\[
\langle A \cLham (1-\Pi) g , g \rangle_{L^2(\mu)} = -\langle (1-\Pi)g , \cLham A^* \Pi g \rangle_{L^2(\mu)}. 
\]
By Lemma~\ref{lem:Lham_A*} below, the operator $\cLham A^*$ is bounded, so that the absolute value of the right-hand side of the above equality is bounded by $\| \cLham A^* \| \| (1-\Pi)g\|_{L^2(\mu)} \| \Pi g\|_{L^2(\mu)}$.

Gathering all estimates, we obtain 
\[
\langle\langle \cL^* g,g \rangle\rangle \leq -G^T S G, 
\]
with
\[
G = \begin{pmatrix} \| \Pi g\|_{L^2(\mu)} \\ \| (1-\Pi) g\|_{L^2(\mu)} \end{pmatrix}, 
\qquad 
S = \begin{pmatrix} a & b/2 \\ b/2 & c \end{pmatrix},
\]
where
\[
%a = \eps \frac{K^2_\nu \| \nabla U \|^2_{L^2(\kappa)}}{1 + K^2_\nu \| \nabla U \|^2_{L^2(\kappa)}},
a = \eps \frac{\alpha K_\nu^2 }{1 + \alpha K_\nu^2 },
\qquad 
b = - \eps \left( \| \cLham A^* \| + \gamma \| A \cLFD \| \right),
\qquad 
c = \frac{\gamma K^2_\kappa}{\beta} - \eps.
\]
Proposition~\ref{prop:coercivity_scrD} follows provided the smallest eigenvalue of~$S$, namely
\begin{equation}
\label{eq:lambda}
\lambda_\eps(S) = \frac{a + c}{2} - \frac12 \sqrt{(a-c)^2 + b^2},
\end{equation}
is positive. A simple argument shows that this holds true when $\eps$ is of the order of $\min(\gamma,1/\gamma)$, in which case $\lambda_\eps(S)$ is also of the same order of magnitude.

It remains to prove the following lemmas.

\begin{lemma}
  \label{lem:Lham_A*}
  The operator $\cLham A^* \Pi = \cLham^2 \Pi (1-\Pi \cLham^2 \Pi)^{-1}$ is bounded.
\end{lemma}

\begin{proof}
The action of $\cLham^2 \Pi$ is 
\[
\cLham^2 \Pi g = \nabla U^T (\nabla_q^2 \Pi g) \nabla U - \nabla V^T (\nabla^2 U) \nabla_q \Pi g.
\]
Since
\[
\int_{\R^d} \partial^2_{p_i,p_j}U \, d\kappa = \beta \int_{\R^d} (\partial_{p_i} U) (\partial_{p_j}U) \, d\kappa,
\]
a simple computation shows that $\Pi \cLham^2 \Pi$ is the generator of an overdamped Langevin process:
\begin{equation}
  \label{eq:ovd_M}
  \Pi \cLham^2 \Pi g = \cLovd^\mathcal{M} \Pi g, \qquad \cLovd^\mathcal{M} = - \nabla V^T \cM \nabla_q + \frac1\beta \cM : \nabla^2_q g = -\frac{1}{\beta} \nabla_q^* \cM \nabla_q,
\end{equation}
where $A:B = \mathrm{Tr}(A^TB)$ is the contraction of two square matrices.
The action of $\cLham A^* \Pi$ is therefore
\[
\cLham A^* \Pi \psi(q,p) = \nabla U(p)^T (\nabla_q^2 \Pi \psi)(q) \nabla U(p) - \nabla V(q)^T (\nabla^2 U)(p) \nabla_q \Pi \psi(q),
\]
with
\[
\psi = \left(1 - \cLovd^\cM\right)^{-1} \Pi g.
\]
The result then easily follows from Lemma~\ref{lem:estimates_spatial_operators} below and the fact that the matrices $\nabla U \otimes \nabla U$ and $\nabla^2 U$ have all their entries in~$L^2(\kappa)$.
\end{proof}

\begin{lemma}
\label{lem:ALFD}
The operator $A \cLFD$ is bounded and
\[
\left| \langle A \cLFD g , g \rangle_{L^2(\mu)} \right| \leq \| A \cLFD \| \|(1-\Pi)g\|_{L^2(\mu)} \|\Pi g\|_{L^2(\mu)}.
\]
\end{lemma}

\begin{proof}
We start by computing the action of $A \cLFD$. Note that $A\cLFD = -(1-\Pi \cLham^2\Pi)^{-1} \Pi \cLham \cLFD = -(1-\Pi \cLham^2\Pi)^{-1} \Pi [\cLham,\cLFD]$ since $\Pi \cLFD = 0$. In order to evaluate the commutator, we compute
\[
-\partial_{q_i}V \partial_{p_i} \cLFD g + \cLFD \left( \partial_{q_i}V \partial_{p_i}g \right) = \partial_{q_i}V \nabla_p \left(\partial_{p_i}U\right)^T \nabla_p g,
\]
and
\[
\begin{aligned}
\partial_{p_i}U \, \partial_{q_i} \cLFD g - \cLFD \left(\partial_{p_i}U \, \partial_{q_i} g \right) & = 
\nabla U^T \nabla_p \left(\partial_{p_i}U\right) \, \partial_{q_i} g \\
& \ \ - \frac{2}{\beta} \nabla_p \left(\partial_{p_i}U\right)^T \nabla_p \partial_{q_i} g - \frac1\beta \partial_{p_i} \left( \Delta U\right) \partial_{q_i} g.
\end{aligned}
\] 
Therefore,
\[
[\cLham,\cLFD]g = \nabla V^T \left(\nabla^2 U\right) \nabla_p g + \nabla U^T \left(\nabla^2 U\right) \nabla_q g - \frac2\beta \nabla^2 U : \nabla^2_{q,p} g - \frac1\beta \nabla (\Delta U)^T \nabla_q g.
\]
We next apply $\Pi$ to the various terms. Since
\[
\begin{aligned}
-\int_{\R^d} \nabla^2 U : \nabla^2_{q,p} g \, d\kappa & = \int_{\R^d} (\nabla_{q} g)^T (\nabla^2 U) \nabla_p \kappa + \int_{\R^d} (\nabla_{q} g)^T \nabla (\Delta U)  \, d\kappa \\
& = -\beta \int_{\R^d} (\nabla_{q} g)^T (\nabla^2 U) \nabla U \, d\kappa + \int_{\R^d} (\nabla_{q} g)^T \nabla (\Delta U)  \, d\kappa,
\end{aligned}
\]
we obtain
\begin{equation}
  \label{eq:Pi_commutator}
  \Pi[\cLham,\cLFD]g = \Pi \left[ \nabla V^T \left(\nabla^2 U\right) \nabla_p g - \left( \left(\nabla^2 U\right)\nabla U - \frac1\beta \nabla (\Delta U)\right)^T \nabla_q g\right].
\end{equation}
Therefore, $T = A\cLFD = (1-\cLovd^\mathcal{M})^{-1}\Pi\mathcal{A}$ with 
\[
\mathcal{A} g = \mathcal{A}_1 g - \mathcal{A}_2 g, 
\qquad 
\mathcal{A}_1 g = \nabla V^T \left(\nabla^2 U\right) \nabla_p g,
\qquad 
\mathcal{A}_2 g = \left( \left(\nabla^2 U\right)\nabla U - \frac1\beta \nabla (\Delta U)\right)^T \nabla_q g. 
\]
Both operators $T_i = (1-\cLovd^\mathcal{M})^{-1}\Pi\mathcal{A}_i$ for $i \in \{1,2\}$ are the composition of two bounded operators, one acting on the position variables only and the other one acting on the momentum variables only; see Lemmas~\ref{lem:estimates_spatial_operators} and~\ref{lem:Pi_U_dp} below. Therefore, $T=T_1+T_2$ is bounded on~$L^2(\mu)$.

To conclude the proof, we note that $\langle A \cLFD g , g \rangle_{L^2(\mu)} = \langle \Pi A \cLFD (1-\Pi) g , g \rangle_{L^2(\mu)} = \langle A \cLFD (1-\Pi) g , \Pi g \rangle_{L^2(\mu)}$. The desired bound then follows from a Cauchy--Schwarz inequality.
\end{proof}

\begin{lemma}
  \label{lem:estimates_spatial_operators}
  Assume that~\eqref{eq:regularization condition} holds. Then, the operators $\partial^2_{q_i,q_j} \left(1 - \cLovd^\cM\right)^{-1}$, $(\partial_{q_i} V) \partial_{q_j} (1-\cLovd^\mathcal{M})^{-1}$ (for any $i,j \in \{1,\dots,d\}$), $(1-\cLovd^\mathcal{M})^{-1} \nabla_q$ and $(1-\cLovd^\mathcal{M})^{-1} |\nabla V|$ are bounded on~$L^2(\nu)$.
\end{lemma}

\begin{proof}
  The condition~\eqref{eq:regularization condition} ensures that the operator $(1+\nabla_q^*\nabla_q)^{-1}$ is bounded from $L^2(\nu)$ to $H^2(\nu)$ (see~\cite{DMS15}). Since $\cM$ is positive definite, there exists $K \geq 1$ such that (in the sense of positive self-adjoint operators)
\[
\frac1K \nabla_q^* \nabla_q \leq -\cLovd^\cM \leq K \nabla_q^* \nabla_q,
\]
and so
\begin{equation}
  \label{eq:equivalence_L_ovd_M}
  \frac1K \left(1 + \nabla_q^* \nabla_q\right)^{-1} \leq \left(1 - \cLovd^\cM\right)^{-1} \leq K \left(1 + \nabla_q^* \nabla_q \right)^{-1}.
\end{equation}
Therefore, the operator $\left(1 - \cLovd^\cM\right)^{-1}$ is also bounded from $L^2(\nu)$ to $H^2(\nu)$. This already shows that $\partial^2_{q_i,q_j} \left(1 - \cLovd^\cM\right)^{-1}$ is bounded on~$L^2(\nu)$ for any $i,j \in \{1,\dots,d\}$.

We next use~\cite[Lemma~A.24]{Villani}: there exists $C>0$ such that, for a function $h : \mathcal{D} \to \R$,
\begin{equation}
\label{eq:control_nabla_V}
\| |\nabla V| h\|^2_{L^2(\nu)} \leq C \left( \|h\|_{L^2(\nu)}^2 + \|\nabla_q h\|_{L^2(\nu)}^2 \right).
\end{equation}
This immediately shows that $(\partial_{q_i} V) \partial_{q_j} (1-\cLovd^\mathcal{M})^{-1}$ is bounded and
\[
\left\| (\partial_{q_i} V) \partial_{q_j} (1-\cLovd^\mathcal{M})^{-1} \right\| \leq C\left(\left\| \partial_{q_j} (1-\cLovd^\mathcal{M})^{-1} \right\| + \left\| \nabla_q \partial_{q_j} (1-\cLovd^\mathcal{M})^{-1} \right\|\right),
\]
the two operators on the right-hand side being bounded since $\left(1 - \cLovd^\cM\right)^{-1}$ is bounded from $L^2(\nu)$ to $H^2(\nu)$. 

Moreover, \eqref{eq:control_nabla_V} shows that the operator $|\nabla V| (1+\nabla_q^*\nabla_q)^{-1}$ is bounded on $L^2(\nu)$, with $\| |\nabla V| (1+\nabla_q^*\nabla_q)^{-1}\|^2 \leq C$. The same conclusion holds for its adjoint $(1+\nabla_q^*\nabla_q)^{-1}|\nabla V|$. We can finally conclude that $(1-\cLovd^\mathcal{M})^{-1} |\nabla V|$ is bounded on $L^2(\mu)$ in view of~\eqref{eq:equivalence_L_ovd_M}.

Finally, using the above arguments, the operator $\nabla_q^* (1-\cLovd^\mathcal{M})^{-1} = (-\nabla_q + \beta \nabla V)(1-\cLovd^\mathcal{M})^{-1}$ is bounded on~$L^2(\nu)$, and so is its adjoint $(1-\cLovd^\mathcal{M})^{-1} \nabla_q$.
\end{proof}

\begin{lemma}
\label{lem:Pi_U_dp}
The operators $\Pi (\partial^\alpha U) \partial_p^{\alpha'}$ are bounded on $L^2(\mu)$ for any $\alpha,\alpha' \in \mathbb{N}^d$ with $|\alpha'| \leq 1$, and
\[
\left\| \Pi (\partial^\alpha U)  \right\| \leq \left\| \partial^\alpha U \right\|_{L^2(\kappa)},
\qquad 
\left\| \Pi (\partial^\alpha U) \partial_p^{\alpha'} \right\| \leq \left\| \partial^{\alpha+\alpha'} U \right\|_{L^2(\kappa)} + \beta \left\| (\partial^{\alpha} U)(\partial^{\alpha'} U) \right\|_{L^2(\kappa)}.
\]
\end{lemma}

\begin{proof}
Let us start with the case $\alpha' = 0$. For $g \in L^2(\mu)$,
\[
\left(\Pi (\partial^\alpha U) g\right)(q) = \int_{\R^d} \partial^\alpha U(p) g(q,p) \, \kappa(dp),
\]
so that, by a Cauchy--Schwarz inequality with respect to the measure~$\kappa$ and a subsequent integration with respect to~$\nu$,
\[
\left\| \Pi (\partial^\alpha U) g \right\|_{L^2(\mu)} \leq \left\| \partial^\alpha U \right\|_{L^2(\kappa)} \|g\|_{L^2(\mu)}.
\]
For the case $|\alpha'| = 1$, we note that
\[
\begin{aligned}
\left(\Pi (\partial^\alpha U) \partial_p^{\alpha'}g\right)(q) & = \int_{\R^d} \partial^\alpha U(p) \partial_p^{\alpha'}g(q,p) \, \kappa(dp) \\
& = -\int_{\R^d} g(q,p) \partial^{\alpha+\alpha'} U(p) \, \kappa(dp) + \beta \int_{\R^d} g(q,p) \partial^{\alpha} U(p) \partial^{\alpha'} U(p) \, \kappa(dp),
\end{aligned}
\]
so that, again by a Cauchy--Schwarz inequality,
\[
\left\| \Pi (\partial^\alpha U) \partial_p^{\alpha'}g \right\|_{L^2(\mu)} \leq \left( \left\| \partial^{\alpha+\alpha'} U \right\|_{L^2(\kappa)} + \beta \left\| (\partial^{\alpha} U)(\partial^{\alpha'} U) \right\|_{L^2(\kappa)}\right) \|g\|_{L^2(\mu)},
\]
which gives the desired conclusion.
\end{proof}

We conclude this appendix with a discussion on the dependence of the convergence rate on the kinetic energy~$U$.

\begin{remark}
\label{rmk:dep_U}
A lower bound on the convergence rate is given by~\eqref{eq:lambda}. There are various places where the kinetic energy~$U$ enters: in a rather explicit way in the coefficients $a$ and~$c$ through the Poincar\'e constant $K_\kappa$ and the term $\|\nabla U\|_{L^2(\mu)}$; but also in a quite cumbersome manner in the operator norms $\| A \mathcal{L}_{\rm FD}\|$ and $\|A^* \mathcal{L}_{\rm ham}\|$, see the proofs of Lemmas~\ref{lem:Lham_A*} and~\ref{lem:ALFD}. It is therefore difficult with our proof to quantify precisely how the lower bound~\eqref{eq:lambda} depends on~$U$.
\end{remark}

%---------------- REMERCIEMENTS ------------------
\subsection*{Acknowledgments}

We would like to thank to Sam Livingstone and Nawaf Bou-Rabee for fruitful discussions. The work of Gabriel Stoltz was funded by the Agence Nationale de la Recherche, under grant ANR-14-CE23-0012 (COSMOS). He also benefited from the scientific environment of the Laboratoire International Associ\'e between the Centre National de la Recherche Scientifique and the University of Illinois at Urbana-Champaign. Zofia Trstanova gratefully acknowledges funding from the European Research Council through the ERC StartingGrant No. 307629 and EPSRC grant EP/P006175/1. 

%--------------- -BIBLIO -----------------
\bibliographystyle{plain}
%\bibliography{MyBibliography}

%\bibliography{MyBibliography}

\end{document}